\documentclass[11pt]{article}
\newcommand{\footremember}[2]{%
    \footnote{#2}
    \newcounter{#1}
    \setcounter{#1}{\value{footnote}}%
}
\newcommand{\footrecall}[1]{%
	    \footnotemark[\value{#1}]%
} 
\title{Characterization and Lower Bounds for Branching Program Size using Projective Dimension}
\author{Krishnamoorthy Dinesh\footremember{a}{Indian Institute of Technology Madras, Chennai, India.
(\texttt{\{kdinesh,sajin,jayalal\}@cse.iitm.ac.in})} \and Sajin Koroth\footrecall{a} \and  Jayalal Sarma\footrecall{a}
}

\usepackage[citecolor=blue,linkcolor=blue,colorlinks=true]{hyperref}
\usepackage{amsthm,amssymb,amsmath}
\usepackage[disable]{todonotes}
\usepackage{complexity}
\usepackage{comment}
\usepackage{enumerate}
\usepackage{fullpage}
\usepackage[capitalise]{cleveref}
\usepackage{collect}

\theoremstyle{plain}% default
\newtheorem{theorem}{Theorem}[section]
\newtheorem{property}[theorem]{Property}
\newtheorem{lem}[theorem]{Lemma} 

\newtheorem{proposition}[theorem]{Proposition}
\newtheorem{claim}[theorem]{Claim}

\theoremstyle{remark}
\newtheorem{remark}[theorem]{Remark}

\theoremstyle{definition}
\newtheorem{definition}[theorem]{Definition}

\newcommand{\pd}{\mbox{{\sf pd}}} 

\newcommand{\pdim}{\pd}

\newcommand{\spd}{\mbox{{\sf spd}}}

\newcommand{\upd}{\mbox{{\sf upd}}}
\newcommand{\uspd}{\mbox{{\sf uspd}}}

\newcommand{\bpdim}{\mbox{{\sf bitpdim}}}

\newcommand{\DISJ}{\mbox{\sf DISJ}}

\newcommand{\INEQ}{\mbox{\sf INEQ}}
\newcommand{\EQ}{\mbox{\sf EQ}}
\newcommand{\ED}{\mbox{\sf ED}}
\newcommand{\PAL}{\mbox{\sf PAL}}

\newcommand{\SI}{{\sf SI}}

\newcommand{\gaussian}[3]{ {\genfrac{[}{]}{0pt}{}{#1}{#2}}_{#3}}

\newcommand{\field}{\mathbb{F}}
\newcommand{\F}{{\mathbb{F}}}

\newcommand{\N}{{\mathbb{N}}}

\renewcommand{\R}{{\mathbb{R}}}
\newcommand{\rank}[1]{\mathsf{rank} \left( #1 \right)}

\newcommand{\rowspan}{\mathsf{rowspan}}
\newcommand{\bpsize}{\mathsf{bpsize}}
\renewcommand{\dim}{\mathsf{dim}}

\newcommand{\set}[1]{\left\{ #1 \right\}}
\newcommand{\vsp}[1]{\mathsf{span} \set{#1}}
\newcommand{\vspan}[2]{\underset{#1}{\mathsf{span}}\{#2\} }
\newcommand{\zo}{\set{0,1}}
\newcommand{\zon}{\zo^n}

\newcommand{\pdrl}{Pudl\'ak-R\"odl }

\newcommand{\bp}{\mathsf{bp}}
\newcommand{\bc}{\mathsf{bc}}

\renewcommand{\D}{\mathsf{D}}

\newcounter{todocounter}
\newcounter{revcounter}

\newcommand{\revnum}[2][]{\stepcounter{revcounter}\todo[#1]{\therevcounter: #2}}
\newcommand{\revsay}[1]{\revnum[color=cyan!20]{\small #1}}

\newcommand{\family}[1]{\mathcal{#1}}

\newcommand{\calB}{{\cal B}}
\newcommand{\calC}{{\cal C}}
\newcommand{\calD}{{\cal D}}

\newcommand{\calF}{{\cal F}}
\newcommand{\calG}{{\cal G}}
\newcommand{\calH}{{\cal H}}

\newcommand{\calP}{{\cal P}}

\def\movetoappendix{0}

\definecollection{appendix}
\makeatletter
\newenvironment{aproof}[2]
  { \@nameuse{collect}{appendix}
  { \subsection{#1} \label{#2} \begin{proof} } {\end{proof}}
  }{\@nameuse{endcollect}}
\makeatother

\ifthenelse{\equal{\movetoappendix}{0}}{
        \renewenvironment{aproof}[2]{\begin{proof}} {\end{proof} }
}{}

\makeatletter
\newenvironment{atheoremproof}[2]
  { \@nameuse{collect}{appendix}
  { \subsection{#1} \label{#2} }
  {}
  }{\@nameuse{endcollect}}
\makeatother

\ifthenelse{\equal{\movetoappendix}{0}}{
        \renewenvironment{atheoremproof}[2]{} {}
}{}

\begin{document}
\maketitle
\begin{abstract}
	We study projective dimension, a graph parameter (denoted by
	$\mathsf{pd}(G)$ for a graph $G$), introduced by Pudl\'ak and R\"odl
	(1992). For a Boolean function $f$(on $n$ bits), Pudl\'ak and R\"odl
	associated a bipartite graph $G_f$ and showed that size of the optimal
	branching program computing $f$ (denoted by $\mathsf{bpsize}(f)$) is
	at least $\mathsf{pd}(G_f)$ (also denoted by $\mathsf{pd}(f)$). Hence,
	proving lower bounds for $\mathsf{pd}(f)$ imply lower bounds for
	$\mathsf{bpsize}(f)$. Despite several attempts (Pudl\'ak and R\"odl
	(1992),  R\'onyai et.al, (2000)), proving super-linear lower bounds
	for projective dimension of explicit families of graphs has remained
	elusive. 

	We observe that there exist a Boolean function $f$ for which the gap
	between the $\mathsf{pd}(f)$ and {\sf bpsize}$(f)$) is
	$2^{\Omega(n)}$.  Motivated by the argument in Pudl\'ak and R\"odl
	(1992), we define two variants of projective dimension -
	\textit{projective dimension with intersection dimension 1} (denoted
	by $\mathsf{upd}(f)$) and \textit{bitwise decomposable projective
	dimension} (denoted by $\mathsf{bpdim}(f)$). We show the following
	results : 
\begin{enumerate}
	\item[(a)] We observe that there exist a Boolean function $f$ for which the gap
	between $\mathsf{upd}(f)$ and {\sf bpsize}$(f)$ is $2^{\Omega(n)}$.  In
	contrast, we also show that the bitwise decomposable projective
	dimension characterizes size of the branching program up to a
	polynomial factor. That is, there exists a constant $c>0$ and
	for any function $f$, \revsay{removed "large" from "large constant"}
		\begin{center}$\mathsf{bpdim}(f)/6 \le \textrm{\sf bpsize}(f)
		\le (\mathsf{bpdim}(f))^c$\end{center}
	\item[(b)] We introduce a new candidate function family $f$ for showing
	super-polynomial lower bounds for $\mathsf{bpdim}(f)$. As our main
	result, we demonstrate gaps between $\mathsf{pd}(f)$ and the
	above two new measures for $f$ : 
		\begin{center}
		$
		\textrm{$\mathsf{pd}(f) = O(\sqrt{n})$ \hspace{5mm}
		$\mathsf{upd}(f) = \Omega(n)$ \hspace{5mm} $\mathsf{bpdim}(f) 
		= \Omega\left(\frac{n^{1.5}}{\log n}\right)$}
		$
		\end{center}
\item[(c)] Although not related to branching program lower bounds, we derive
	exponential lower bounds for two restricted variants of
	$\mathsf{pd}(f)$ and $\mathsf{upd}(f)$ respectively by observing that
	they are exactly equal to well-studied graph parameters - bipartite
	clique cover number and bipartite partition number respectively.
\end{enumerate}
\end{abstract}

\section{Introduction}
A central question in complexity theory - the $\P$ vs $\L$ problem - asks if a deterministic Turing machine that runs in polynomial
time can accept any language that cannot be accepted by deterministic Turing machines with logarithmic space bound. A stronger
version of the problem asks if $\P$ is separate from $\L/\poly$ (deterministic logarithmic space given polynomial sized advice). 
The latter, recast in the language of circuit complexity theory, asks if there exists an \textit{explicit} family of functions ($f:\{0,1\}^n \to \{0,1\}$) computable in polynomial time (in terms of $n$), such that any family of deterministic branching programs computing them has to be of size $2^{\Omega(n)}$. 
However, the best known non-trivial size lower bound against deterministic branching programs, due to Nechiporuk
\cite{Nec66} in 1970s, is $\Omega(\frac{n^2}{\log^2 n})$. 

Pudl\'ak and R\"odl \cite{PR92} described a linear algebraic approach to show size lower bounds against deterministic branching programs. They introduced a linear algebraic parameter called \emph{projective dimension} (denoted by $\pd_{\F}(f)$, over a field $\F$) defined on a natural graph associated with the Boolean function $f$. 
For a Boolean function $f:\{0,1\}^{2n} \to \{0,1\}$, fix a partition of the input bits into two parts of size $n$ each, and consider the bipartite graph $G_f(U,V,E)$ defined on vertex sets $U = \{0,1\}^n$ and $V = \{0,1\}^n$, as $(u,v) \in E$ if and only if $f(uv) = 1$. We call $G_f$ as the \emph{bipartite realization} of $f$. For a bipartite graph $G(U,V,E)$, the projective dimension of $G$ over a field $\F$, denoted by $\pd_{\field}(G)$, is defined as the smallest $d$ for which there is a vector space $W$ of dimension $d$ (over $\F$) and a function $\phi$ mapping vertices
in $U, V$ to linear subspaces of $W$ such that for all $(u,v) \in U \times V$,
$(u,v) \in E$ if and only if $\phi(u) \cap \phi(v) \ne \{0\}$. We say that $\phi$ {\em realizes} the graph $G$.

Pudl\'ak and R\"odl \cite{PR92} showed that if $f$  can be computed by a
deterministic branching program of size $s$, then $\pd_\F(f) \le s$ over any field $\F$.
Thus, in order to establish size lower bounds against branching
programs, it suffices to prove lower bounds for projective dimension
of explicit family of Boolean functions. 

\revsay {The lower bound over reals is using a bound of Warren. So removed "By a counting argument"}
Pudl\'ak and R\"odl in \cite{PR92} showed that for most Boolean functions $f:\{0,1\}^n \times \{0,1\}^n \to \{0,1\}$, $\pd_{\mathbb{R}}(f)$ is $\Omega(\sqrt{\frac{2^n}{n}})$. In a subsequent work, the same authors~\cite{PR94} also established an upper bound $\pd_{\mathbb{R}}(f) =  O(\frac{2^n}{n})$ for all functions. More recently, R\'onyai, Babai and Ganapathy~\cite{BRG00}
established the same lower bound over all fields. 
Over finite fields $\F$, Pudl\'ak and R\"odl \cite{PR92} also showed 
(by a counting argument) that there exists a Boolean function $f:\{0,1\}^n \times \{0,1\}^n \to \{0,1\}$ such that $\pd_{\F}(f)$ is $\Omega(\sqrt{2^n})$. 
However, till date, obtaining an explicit family of Boolean functions (equivalently graphs) achieving such 
lower bounds remain elusive. 
The best lower bound for projective dimension for an explicit family of functions is for the inequality function (on $2n$ bits, the graph is the bipartite complement of the perfect matching) where a lower bound of $\epsilon n$ for an absolute constant $\epsilon > 0$ is known~\cite{PR92} over $\R$. For a survey on
projective dimension and related linear algebraic techniques,
refer~\cite{PR94,lokam2009complexity}. 
Thus, the best known size lower bound that was
achieved using this framework is only $\Omega(n)$ which is not better than trivial lower bounds.

\noindent {\bf Our Results :}
The starting point of our investigation is the observation that projective assignment appearing 
in the proof of \cite{PR92} also has the property that the dimension of the intersection of 
two subspaces assigned to the vertices is exactly $1$, whenever they intersect (See Proposition~\ref{prop:pdrl-props}(2)). We denote, for 
a function $f$, the variant of projective dimension defined by this property as 
$\upd(f)$ (See Section~\ref{sec:updim}).
From the above discussion, for any Boolean function $f$, $\pd(f) \le \upd(f) \le \bpsize(f)$.  
A natural question is whether this restriction helps in proving better lower bounds 
for the branching programs. By observing properties about \revsay{Changed "the measure" to "this measure"} projective dimension and choosing a new candidate function\footnote{the candidate function is in $\NC^2$ but unlikely to be in $\NL$. See~\cref{th:pd-hardness}.}, we demonstrate that the above restriction can help by proving the following quadratic gap between the two measures.
\begin{theorem}\label{thm:pd-lb-main}
For any $d \ge 0$, for the function $\SI_d$ (on $2d^2$ variables, See Definition \ref{def:sid}), the projective dimension is exactly equal to $d$, while the projective dimension with intersection dimension $1$ is $\Omega(d^2)$.
\end{theorem}
However, this does not directly improve the known branching program size lower
bound for $\SI_d$, since it leads to only a linear lower bound on
$\upd(\SI_d)$. We demonstrate the weakness of this measure
by showing the existence of a function (although not explicit) for which there is an exponential gap between $\upd$ {\em over any partition} and the branching program size~(Proposition~\ref{thm:upd-failure}). 
This motivates us to look for variants of projective dimension of graphs, which is 
closer to the optimal branching program size of the corresponding Boolean 
function. We observe more properties (see 
Proposition~\ref{prop:pdrl-props}) about the subspace assignment 
from the proof of the upper bound from \cite{PR92}.
 We call the projective assignments with 
these properties {\em bitwise decomposable projective assignment} and 
denote the corresponding dimension\footnote{We do not use 
the property that intersection dimension is $1$ and hence is 
incomparable with $\upd$.} 
as $\bpdim(f)$~(See Definition~\ref{def:bitpdim}). Thus, for any Boolean 
function $f$, $\pd(f) \le \bpdim(f)$. We also show that $\bpdim(f) \le 6\cdot\bpsize(f)$ (Lemma \ref{bpdim-pdrl}). To demonstrate the 
tightness of the definition, we first argue a converse with 
respect to this new parameter.
\begin{theorem} \label{thm:bitpdim-converse}
There is an absolute constant $c>0$ such that if $\bpdim(f_n) \le d(n)$ for 
a function family $\{f_n\}_{n \ge 0}$ on $2n$ bits, then there 
is a deterministic branching program of size $(d(n))^c$ computing it.
\end{theorem}

Thus, super-polynomial size lower bounds for branching programs imply
super-polynomial lower bounds for $\bpdim(f)$. The function $\SI_d$ (on $2d^2$
input bits - See Definition~\ref{def:sid}) is a natural candidate for proving
$\bpdim$ lower bounds as the corresponding language is hard\footnote{Assuming
$\CeL \not\subseteq \L/\poly$, $\SI_d$ cannot be
computed by deterministic branching programs of size $\poly(d)$.} for the
complexity class $\CeL$ under logspace Turing reductions.

However, the best known lower bound for branching program size for an explicit family of functions is $\Omega\left (\frac{n^2}{\log^2 n} \right )$ by Nechiporuk~\cite{Nec66} which uses a counting argument on the number of sub-functions. By Theorem~\ref{thm:bitpdim-converse} , $\bpdim(f)$ (for the same explicit function) is 
at least $\Omega\left (\frac{n^{2/c}}{\log^{2/c} n} \right )$. The constant 
$c$ is more\footnote{The value of $c$ can be shown to be at most $3+\epsilon$. See proof of Theorem~\ref{thm:bitpdim-converse} in Section~\ref{sec:bpsize-char}.}  \revsay{"large" to "more than 2"}  than $3$
and hence implies only weak lower bounds for $\bpdim$. 
Despite this weak connection, by combining the counting strategy with the linear algebraic structure of $\bpdim$, we show a super-linear lower bound for $\SI_d$ matching the branching program size lower bound\footnote{A lower bound of $\Omega\left (\frac{d^{3}}{\log d} \right )$ for the branching program size can also be obtained using Nechiporuk's method.}.

\begin{theorem}[Main Result]\label{thm:bitpdim-lb}
For any $d > 0$, $\bpdim(\SI_d)$ is at least $\Omega\left (\frac{d^{3}}{\log d} \right )$.
\end{theorem}

Theorems \ref{thm:pd-lb-main} and \ref{thm:bitpdim-lb} demonstrate gaps between the $\pd$ and the new measures considered. In particular, for $n = d^2$,
$\pd(\SI_d) = O(\sqrt{n})$,
$\upd(\SI_d) = \Omega(n)$, and 
$\bpdim(\SI_d) = \Omega\left(\frac{n^{1.5}}{\log n}\right)$.
We remark that Theorem~\ref{thm:bitpdim-lb} implies a size lower bound of $\Omega(\frac{n^{1.5}}{\log n})$ for branching programs computing the function $\SI_d$ (where $n = d^2$). However, note that this can also be derived from Nechiporuk's method. For the Element Distinctness function, the above linear algebraic adaptation of Nechiporuk's method for $\bpdim$ gives $\Omega(\frac{n^2}{\log^2 n})$ lower bounds (for $\bpdim$ and hence for $\bpsize$) which matches with the best lower bound that Nechiporuk's method can derive. 
This shows that our modification of approach in \cite{PR92} can also achieve the best known lower bounds for branching program size.

Continuing the quest for better lower bounds for projective dimension, we study two further restrictions. In these variants of $\pd$ and $\upd$, the subspaces assigned to the vertices must be spanned by standard basis vectors. We denote the corresponding dimensions as $\spd(f)$ and $\uspd(f)$ respectively. It is easy to see that for any $2n$-bit function, both of these dimensions are upper bounded by $2^n$.%(Proposition~\ref{prop:}).

We connect these variants to some of the well-studied graph parameters. The \textit{bipartite clique cover number} (denoted by $\bc(G)$) is the smallest collection of complete bipartite subgraphs of $G$ such that every edge in $G$ is present in some graph in the collection. If we insist that the bipartite graphs in the collection be edge-disjoint, the measure is called \textit{bipartite partition number} denoted by $\bp(G)$. By definition, $\bc(G) \le \bp(G)$. These graph parameters are closely connected to communication complexity as well. More precisely, $\log(\bc(G_f))$ is exactly the non-deterministic communication complexity of the function $f$, and $\log(\bp(G_f))$ is a lower bound on the deterministic communication complexity of $f$~(See \cite{juknatext}). In this context, we show the following:
\begin{theorem}\label{th:spdEqbc}
For any Boolean function $f$,  $\spd(f) = \bc(G_f) $ and $\uspd(f) = \bp(G_f)$.
\end{theorem}
Thus, if for a function family, the non-deterministic communication complexity is $\Omega(n)$, then we will have $\spd(f) = 2^{\Omega(n)}$. Thus, both $\spd(\DISJ)$ and $\uspd(\DISJ)$ are $2^{\Omega(n)}$.

\section{Preliminaries}
\label{sec:prelims}
In this section, we introduce the notations used in the paper. 
For definitions of basic complexity classes and computational models, we refer the reader to standard textbooks~\cite{juknatext,vollmertext}. 

Unless otherwise stated, we work over the field $\F_2$. We remark that our
arguments do generalize to any finite field. All subspaces that
we talk about in this work are linear subspaces. Also $\vec{0}$ and $\set{0}$
denote the zero vector, and zero-dimensional space respectively.
For a subspace $U \subseteq \F^n$, we call the ambient dimension of $U$ as $n$. We denote $e_i \in \F^n$ as the $i^{th}$ standard basis vector with $i^{th}$ entry being $1$ and rest of the entires being zero.

For a graph $G(U,V,E)$, recall the definition of projective dimension of $G$ over a field $\F$($\pd_{\field}(G)$), defined in the introduction. 
For a Boolean function $f:\{0,1\}^{2n} \to \{0,1\}$, fix a partition of the input bits into two parts of size $n$ each, and consider the bipartite graph $G_f$ defined on vertex sets $U = \{0,1\}^n$ and $V = \{0,1\}^n$, as $(u,v) \in E$ if and only if $f(uv) = 1$.
A $\phi$ is said to \emph{realize} the function $f$ if it {\em realizes} $G_f$.
Unless otherwise mentioned, the partition is the one specified in the definition of the function.
We denote by $\bpsize(f)$ the number of 
vertices (including accept and reject nodes) in the optimal branching program 
computing $f$.

\begin{theorem}[\pdrl Theorem~\cite{PR92}]
For a Boolean function $f$ computed by a deterministic branching program
of size $s$ and $\F$ being any field, $\pd_\F(G_f) \le s$. 
\label{thm:pdrl}
\end{theorem}
The proof of this result
proceeds by producing a subspace assignment for vertices of $G_f$ from a branching program computing $f$. % (See appendix section XX for a proof)
We reproduce the proof of the above theorem in our notation, in Appendix~\ref{app:pdrl-proof} and derive the following proposition from the same.

\begin{proposition}\label{prop:pdrl-props}
For a Boolean function $f :\set{0,1}^n \times \set{0,1}^n \to \set{0,1}$ computed by a deterministic branching program of size $s$, there is a
collection of subspaces of $\F^s$ denoted $\calC = \{U_i^a\}_{i \in
	[n], a \in \{0,1\}}$ and $\calD = \{V_j^b\}_{j \in [n], b \in
	\{0,1\}}$, where we associate the subspace $U_i^a$ with a bit
	assignment $x_i = a$ and $V_j^b$ with $y_j = b$
such that if we define the 
 map $\phi$ assigning subspaces from $\F^s$ to vertices of $G_f(U,V,E)$ as $\phi(x) =
		\vspan{1\le i \le n}{U_i^{x_i}}$, $\phi(y)= \vspan{1\le j \le
		n}{V_j^{y_j}}$, for $x \in X, y\in Y$ then the following holds true. Let $S = \{e_i - e_j \mid i,j \in [s], i \ne j\}$. 
\begin{enumerate}
  \item for all $(u,v) \in U \times V$, $\phi(u) \cap \phi(v) \neq \set{0}$ if and only if $f(u,v)=1$.
  \item for all $(u,v) \in U \times V$, $\dim\left(\phi(u) \cap \phi(v)\right) \le 1$.
  \item For any $W \in \calC \cup \calD$, $\exists S' \subseteq S$ such that $W = \vsp{S'}$.
\end{enumerate}
\end{proposition}

\begin{aproof}{Proof of \cref{prop:pdrl-props}}{app:pdrl-props}
We reuse the notations introduced in proof of Theorem~\ref{thm:pdrl} which we have described in the Appendix~\ref{app:pdrl-proof}.
If $H_x$ denotes the set of edges that are closed on an input $a$, then the subspace assignment $\phi(a)$ is span of vectors associated with edges of $H_x$. Denote by $H_{x_i=a_i}$, the subgraph consisting of edges labeled $x_i=a_i$. Hence $H_a$ can be written as span of vectors associated with $H_{x_i=a_i}$. Hence $\phi(a)$ can be expressed as $span_{i=1}^n U_i$ where $U_i = span_{(u,v) \in H_{x_i=a_i}} (e_u-e_v)$. A similar argument shows that $\phi(y)$ also has such a decomposition. We now argue the properties of $\phi$.

Note that the first and third property directly follow from proof. To see second property, observe that the branching program is deterministic and hence there can be only one accepting path. Since we observed that the vectors in the accepting path contribute to the intersection space and since there is only one such path, dimension of the intersection  spaces is bound to be $1$.
\end{aproof}

We define the following family of functions and family of graphs based on
subspaces of a vector space and their intersections.
\begin{definition}[$\SI_d$, $\calP_d$]
\label{def:sid}
	Let $\F$ be a finite field. Denote by $\SI_d$, the Boolean \revsay{Do we want to call $SI$ Boolean ie is $\F$ fixed as $\F_2$ ?} function defined on $\F^{d \times d}
	\times \F^{d \times d} \to \{0,1\}$ as for any $A,B \in \F^{d
	\times d}$ $\SI_d(A,B) = 1$ if and only if $\rowspan(A) \cap
	\rowspan(B) \ne \set{0}$. Note that the row span is over the field 	$\F$ (which, in our case, is $\F_2$).
	Denote by $\calP_d$, the bipartite graph $(U,V,E)$ where $U$ and $V$
	are the set of all subspaces of $\F^d$. And for any $(I,J) \in U \times
	V$, $(I,J) \in E \iff I \cap J \ne \set{0}$
\end{definition}

We collect the definitions of Boolean functions which we deal with in this work.
For  $(x,y) \in \set{0,1}^n \times \set{0,1}^n$,
$\IP_n(x,y) = \sum_{i=1}^n x_iy_i \mod 2$, 
$\EQ_n(x,y)$ is $1$ if $\forall i \in [n]$ $x_i = y_i$ and is $0$ otherwise,  
$\INEQ_n(x,y) = \neg \EQ_n(x,y)$ and $\DISJ_n(x,y) = 1$ if $\forall i \in [n]$
$x_i \land y_i = 0$ and is $0$ otherwise. Note that all the functions 
discussed so far has branching programs of size $O(n)$ computing them and hence 
have projective dimension $O(n)$ by Theorem~\ref{thm:pdrl}.

Let $m \in \N$ and $n = 2m\log m$. The Boolean function, Element Distinctness, denoted $\ED_n$ is defined 
	on $2m$ blocks of $2\log m$ bits, $x_1,\ldots,x_{m}$ and
	$y_1,\ldots,y_{m}$ bits and it evaluates to $1$ if and only if all the $x_i$s
	and $y_i$s take distinct values when interpreted as integers in $[m^2]$.
Let $q$ be a power of prime congruent to 1 modulo 4. Identify elements in $\{0,1\}^n$ with elements of $\F_q^{*}$. For $x, y \in \F_q^{*}$, the Paley function $\PAL^q_n (x,y) = 1$ if $x-y$ is a quadratic residue in $\F_q^*$ and $0$ otherwise.

We observe for any induced subgraph $H$ of $G$, 
	if $G$ is realized in a space of dimension $d$, then
	$H$ can also be realized in a space of dimension $d$.
For any $d \in \N$, $\calP_d$ appears as an induced subgraph of the bipartite realization of $\SI_d$. Hence, $\pd(\SI_d) \ge \pd(\calP_d)$.

\begin{atheoremproof}{Linear Algebra Basics}{linalg}
We need the following definition of Gaussian coefficients. For non-negative integers $n,k$ and a prime power $q$, $\gaussian{n}{k}{q}$ is the expression,
$\frac{(q^n-1)(q^n-q) \ldots 	(q^n-q^{k-1})}{(q^k-1)(q^k-q) \ldots	(q^k-q^{k-1})} \text{ if } n \ge k,	k \ge 1$, $0$ if $n < k,  k \ge 1$, $1$ if $n \ge 0, k = 0$.

\noindent
\textbf{Linear Algebra :}
We recall some basic lemmas from linear algebra which we use later. Unless otherwise mentioned, all our algebraic formulations are over finite fields ($\mathbb{F}$ of size $q$). For vector spaces $V_1$, $V_2$ with dimensions $k_1$, $k_2$ respectively, the
\emph{direct sum} $V_1 \oplus V_2$ is the vector space formed by the column space of 
the matrix \( M = \begin{bmatrix}
                B_1 & 0  \\
                0 & B_2 
              \end{bmatrix}
        \)
where $B_1$ is a $k_1 \times k_1$ matrix whose column space forms $V_1$, $B_2$
is a $k_2 \times k_2$ matrix whose column space form $V_2$.
We now state a useful property of direct sum.
\begin{proposition}
\label{prop:dir-sum-prop}
For an arbitrary field $\F$, let $U_1$, $V_1$ be subspaces of $\F^{k_1}$ and
$U_2, V_2$ be subspaces of $\F^{k_2}$. Then,
$(U_1 \oplus U_2) \cap (V_1 \oplus V_2) \ne \{0\} \iff U_1 \cap V_1 \ne
\{0\} \text{ or } U_2 \cap V_2 \ne \{0\}$
\end{proposition}

Let $U, V$ be two vector spaces. Then the vector space formed by 
$\text{Span}\left (\{ uv^{\top}~|~u \in U, v \in V \} \right )$
is called the \emph{tensor product} of vector spaces $U,V$ denoted as $U \otimes V$. Here $u, v$ are column vectors.
A basic fact about tensor product that we need is the following. 
Let $U$ be a vector space having basis $u_1,u_2,
\ldots u_k$ and $V$ be a vector space having basis $v_1,v_2,\ldots,v_\ell$ 
over some field $\F$ then, vector space $U \otimes V$ has a basis $B 
=\{u_iv_j^{\top} ~|~i \in \{1,2,\ldots,k\}, j \in \{1,2,\ldots, \ell\}  \}$
where $u, v$ are column vectors.
Hence, for any two vector spaces $U,V$ $dim(U \otimes V) = dim(U) \cdot dim(V)$.

\begin{proposition}
\label{prop:dir-prod-prop}
For an arbitrary field $\F$, let $U_1$, $V_1$ be subspaces of
$\F^{k_1}$ and $U_2, V_2$ be subspaces of $\F^{k_2}$. Then,
$(U_1 \otimes U_2) \cap (V_1 \otimes V_2) \ne \{0\} \iff U_1 \cap V_1 \ne \{0\} \text{ and } U_2 \cap V_2 \ne \{0\}$
\end{proposition}
The proofs of the two Propositions~\ref{prop:dir-sum-prop} and~\ref{prop:dir-prod-prop} are fairly elementary and follows from basic linear algebra. For example Proposition~\ref{prop:dir-prod-prop} follows as an easy corollary from an exercise from~\cite[Chapter 14, exercise 12]{R05}\footnote{There is a typo in the way the exercise is stated in~\cite[Chapter 14, exercise 12]{R05}. For this reason we give a proof of this result in Appendix~\ref{app:andlemma}.}.

	Let $V$ be a finite dimensional vector space. For any $U \subseteq_S
	V$, $V = U \oplus U^{\perp}$. Hence for any $v \in V$ there exists a
	unique $u \in U, w \in U^{\perp}$ such that $v = u+w$. A projection
	map $\Pi_U$ is a linear map defined as $\Pi_U(v) = u$
	where $u$ is the component of $v$ in $U$. For any $A,B \subseteq_S V$ with $A \cap B = \set{0}$, let $V = A+B$.
	Then any vector $w \in V$ can be uniquely expressed as $w =
	\Pi_A(w)+\Pi_B(w)$.
It is easy to see that, for any $A,B \subseteq_S \F^d$, with $A \cap B = \set{0}$, and any $w \in \F^d$, $\Pi_{A+B}(w) = \Pi_A(w) + \Pi_B(w)$. \\

\end{atheoremproof}

\section{Properties of Projective Dimension} 

In this section, we observe properties about projective dimension as a measure of graphs and Boolean functions.
We start by proving closure properties of projective dimension under Boolean operations $\land$ and $\lor$. The proof  is based on direct sum and tensor product of vector spaces.

\begin{lem}
\label{lem:OrLemma}
Let $\mathbb{F}$ be an arbitrary field. For any two functions $f_{1}:\left\{ 0,1\right\} ^{2n}\to\left\{ 0,1\right\}$, $f_{2}:\left\{ 0,1\right\} ^{2n}\to\left\{ 0,1\right\} $,
$\pdim_{\mathbb{F}}\left(f_{1} \lor f_{2}\right)\leq\pdim_{\mathbb{F}}\left(f_{1}\right) + \pdim_{\mathbb{F}}\left(f_{2}\right)$ and 
$\pdim_{\mathbb{F}}\left(f_{1} \land
f_{2}\right)\leq\pdim_{\mathbb{F}}\left(f_{1}\right)\cdot \pdim_{\mathbb{F}}\left(f_{2}\right)$
\end{lem}
\begin{aproof}{Proof of \cref{lem:OrLemma}}{app:OrLemma}
In this proof, for a Boolean $f$ with bipartite representation
$G_f(U,V,E)$ we define the map $\phi$ to be
from $\set{0,1}^n \times \set{0,1}$ where $\phi(u,0)$ denotes
the subspace assigned to $u \in U$ and $\phi(v,1)$ denotes the subspace
assigned to $v \in V$ of $G_f$. Let $f_{1}$ and $f_2$ be of projective dimensions $k_{1}$ and $k_2$ realized by maps 
$\phi_{1} :\left\{ 0,1\right\} ^{n}\times\left\{ 0,1\right\}
\to\mathbb{F}^{k_{1}},\phi_{2}:\left\{ 0,1\right\}^{n}\times\left\{
0,1\right\} \to\mathbb{F}^{k_{2}}$ respectively.
\begin{itemize}
\item From $\phi_1$ and $\phi_2$ we construct a
subspace assignment $\phi:\left\{ 0,1\right\} ^{n}\times\left\{
0,1\right\} \to\mathbb{F}^{k_{1}+k_{2}}$ which realize $f_1 \lor f_2$
thus proving the theorem. \\
The subspace assignment is : for $u \in \{0,1\}^n, \phi(u,0) =
\phi_1(u,0) \oplus \phi_2(u,0)$. Similarly for $v \in \{0,1\}^n,
\phi(v,1) = \phi_1(v,1) \oplus \phi_2(v,1)$.  Now, for $u,v \in
\{0,1\}^n$, if $f(u,v) = 1$ then it must be that $f_1(u,v) = 1$ or
$f_2(u,v) = 1$. Thus either $\phi_1(u,0) \cap \phi_1(v,1) \ne \{0\}$
or $\phi_2(u,0) \cap \phi_2(v,1) \ne \{0\}$. By
Proposition~\ref{prop:dir-sum-prop}, it must be the case that $(\phi_1(u,0)
\oplus \phi_2(u,0))\cap (\phi_1(v,1) \oplus \phi_2(v,1)) \ne
\{0\}$. Hence $\phi(u,0) \cap \phi(v,1) \ne \{0\}$.
The dimension of resultant space is $k_1+k_2$.
\item From $\phi_1$ and $\phi_2$ we construct a
subspace assignment $\phi:\left\{ 0,1\right\} ^{n}\times\left\{
0,1\right\} \to\mathbb{F}^{k_{1}k_{2}}$, realizing $f_1 \land f_2$ thus proving the theorem. 
Consider the following projective dimension assignment $\phi$:
 for $u \in \{0,1\}^n, \phi(u,0) = \phi_1(u,0)
\otimes \phi_2(u,0)$. Similarly for $v \in \{0,1\}^n, \phi(v,1) = \phi_1(v,1)\otimes \phi_2(v,1)$. The proof is similar to the previous case and applying Proposition~\ref{prop:dir-prod-prop}, completes the proof.
\end{itemize}
\end{aproof}
The $\lor$ part of the above lemma was also observed (without proof) in \cite{PR94}. 
 A natural question is whether we can improve any
of the above bounds. In that context, we make the following remarks: (1) the construction for $\lor$ is tight up to constant factors,
(2) we cannot expect a general relation connecting $\pd_\R(f)$ and $\pd_\R(\neg f)$.
\begin{atheoremproof}{Tightness of $\lor$, impossibility of $\neg$ over $\R$}{app:or-tight}

\begin{itemize}
\item We prove that the construction for $\lor$ is tight up to constant
	factors. Assume that $n$ is a multiple of $4$. Consider the functions $f
	(x_1, \dots, x_{\frac{n}{4} }, x_{\frac{n}{2} + 1}, \dots , x_{\frac{3
	n}{4}})$ and $g (x_{\frac{n}{4} + 1}, \dots , x_{\frac{n}{2}},
	x_{\frac{3 n}{4} + 1},
\dots , x_n)$ each of which performs inequality check on the first
$\frac{n}{4}$ and the second $\frac{n}{4}$ variables.
It is easy to see that $f \lor g$ is the inequality function on $\frac{n}{2}$
variables $x_1, \dots , x_{\frac{n}{2}}$ and the next $\frac{n}{2}$
variables $x_{\frac{n}{2} + 1}, \ldots , x_n$.  By the fact that they are
computed by $n$ size branching programs and using Theorem~\ref{thm:pdrl}   
 (\pdrl theorem) we get that 
$\pdim (f) \le n$ and $\pdim (g) \le n$. Hence by Lemma~\ref{lem:OrLemma}, 
$\pdim (f \lor g) \le \pdim(f)+\pdim(g) \le 2 n$. Lower bound on projective
dimension of inequality function comes from~\cite[Theorem~4]{PR92},  
giving $\pdim (f \lor g) \ge \epsilon. \frac{n}{2}$ for an absolute constant
$\epsilon$.  This shows that $\pdim (f \lor g) = \Theta(n)$.
\item A natural idea to improve the upper bound of $\pd(f_1 \land f_2)$ is to prove upper bounds for $\pd(\lnot f)$ in terms of $\pd(f)$. However, we remark that over $\R$, it is known~\cite{PR92} that $\pd_\R(\INEQ_n)$ is $\Omega(n)$ while $\pd_\R(\EQ_n) = 2$. Hence we cannot expect a general relation connecting $\pd_\R(f)$ and $\pd_\R(\neg f)$.
\end{itemize}
\end{atheoremproof}

We now observe a characterization of bipartite graphs having projective dimension at most $d$ over $\F$. 
\begin{atheoremproof}{Distinct Neighborhood Implies Different Subspaces}{app:dis-nbhr}
Let $f:\{0,1\}^n \times \{0,1\}^n \to \{0,1\}$, and $G_f(X,Y,E)$ be its
bipartite realization. Let $\pdim(G_f) = d$. 

\begin{proposition} \label{prop:critical}
For any subspace assignment $\phi$ realizing $G_f$, no two vertices from the same partition whose neighborhoods are different can get the same subspace assignment.
\end{proposition}
\begin{proof}
Suppose there exists $x,x' \in S$ from the same partition, i.e., either $X$ or $Y$,such that $\phi(x) = \phi(x')$. Since $N(x)\neq N(x')$, without loss of generality, there exists $z \in N(x)
\setminus N(x')$. Now since $\phi(x) = \phi(x')$, $x'$ will be made 
adjacent to $z$ by the assignment and hence $\phi$ is no longer a realization
of $G_f$ since $z$ should not have been adjacent to $x'$.
\end{proof}
\end{atheoremproof}
\begin{lem}[Characterization] \label{lem:pdim-char}
	Let $G$ be a bipartite graph with no two vertices having same neighborhood,
$\pdim(G)\le d$ if and only if $G$ is an induced subgraph of $\calP_d$.
\end{lem}
\begin{aproof}{Proof of \cref{lem:pdim-char}}{app:pdim-char}
Suppose $G$ appears as an induced subgraph of $\calP_d$. To argue,
$\pdim(G) \le d$, simply consider the assignment where the subspaces
corresponding to the vertices in $\calP_d$ are assigned to the vertices of
$G$. 

On the other hand, suppose $\pdim(G) \le d$. Let $U_1,\ldots,U_N$
and $V_1,\ldots,V_N$ be subspaces assigned to the vertices. Since the
neighborhoods of the associated vertices are different, by
Proposition~\ref{prop:critical}, no two subspaces assigned to these vertices
can be the same. Hence corresponding to each vertex in $G$, there
is a unique vertex in $\calP_d$ which corresponds to the assignment. 
Now the subgraph induced by the vertices corresponding to
these subspaces in $\calP_d$ must be isomorphic to $G$ as the subspace
assignment map for $G$ preserves the edge non-edge relations in $G$.
\end{aproof}
It follows that 
$\pdim(\calP_d) \le d$. Observe that, in any projective assignment, the vertices with
different neighborhoods should be assigned different subspaces.
For $\pdim(\calP_d)$, all vertices on either partitions have distinct neighborhoods. The number of subspaces of a vector space of dimension $d-1$ is strictly smaller than the number of vertices in $\calP_d$. Thus, we conclude the following theorem.
\begin{theorem} \label{thm:subspace-graph-pd}
	For any $d \in \N$, $\pdim(\calP_d) = \pd(\SI_d)= d$.
\end{theorem}
For an $N$ vertex graph $G$, the number of vertices of distinct neighborhood can at most be 
$N$. Thus, the observation that we used to show the lower bound for the 
graph $\pdim(\calP_d)$ cannot be used to obtain more than a $\sqrt{\log 
N}$ lower bound for $\pdim(G)$. Also, for many functions, the number of vertices of distinct neighborhood can be smaller. 

We observe that by incurring an additive factor of $2\log N$, any graph $G$ on $N$ vertices can be transformed into a graph $G'$ on $2N$ vertices such that all the neighborhoods of vertices in one partition are all distinct.
Let $f:\set{0,1}^{2n} \to \set{0,1}$ be such that the neighborhoods of $G_f$
are not necessarily distinct. We consider a new function $f'$ whose bipartite 
realization will have two copies of $G_f$ namely $G_1(A_1,B_1,E_1)$ and 
$G_2(A_2,B_2,E_2)$ where $A_1,A_2,B_1,B_2$ are disjoint and a matching connecting vertices in $A_1$ to $B_2$ and $A_2$ to $B_1$ respectively. Since the matching edges associated with every vertex is unique, the neighborhoods of all 
vertices are bound to be distinct.
Applying Lemma~\ref{lem:OrLemma} and observing that matching (i.e,  
equality function) has projective dimension at most $n$, $\pd(f') \le 2\pd(f) + 2n$. This shows that to show super-linear lower bounds on projective dimension for $f$ where the neighborhoods may not be distinct, it suffices to show a super-linear lower bound for $f'$.

\section{Projective Dimension with Intersection Dimension 1}
\label{sec:updim}
Motivated by the proof of Theorem~\ref{thm:pdrl} (presented in  Appendix~\ref{app:pdrl-proof}) we make the following definition.

\begin{definition}[\textbf{Projective Dimension with Intersection Dimension 1}]
  \label{defn:upd}
  A Boolean function $f : \set{0,1}^n \times \set{0,1}^n \to
  \set{0,1}$ with the corresponding bipartite graph $G(U,V,E)$ is said
  to have projective dimension with intersection dimension $1$ (denoted by $\upd(f)$) $d$ over field $\field$, if $d$ is the smallest possible dimension for which there exists a vector space $K$ of dimension $d$ over $\field$ with a map $\phi$ assigning subspaces of $K$ to $U \cup V$  such that
  \begin{itemize}
  \item for all $(u,v) \in U \times V$, $\phi(u) \cap \phi(v) \neq \set{0}$ if and only if $(u,v) \in E$.
  \item for all $(u,v) \in U \times V$, $\dim\left(\phi(u) \cap \phi(v)\right) \le 1$.
\end{itemize}   
\end{definition}

By the properties observed in Proposition~\ref{prop:pdrl-props},
\begin{theorem}
\label{thm:pdrl-upd}
For a Boolean function $f$ computed by a deterministic branching program
of size $s$, $\upd_\F(f) \le s$ for any field $\F$. 
\end{theorem}

Thus, it suffices to prove lower bounds for $\upd(f)$ in order to obtain 
branching program size lower bounds. 
We now proceed to  show lower bounds on $\upd$.

Our approaches use the fact that the adjacency matrix of $\calP_d$ has high rank.

\begin{lem}
\label{th:subspace-graph-rank}
Let $M$ be the bipartite adjacency matrix of $\calP_d$, then $\rank{M} \ge \gaussian{d}{d/2}{q} \ge q^{\frac{d^2}{4}}$
\end{lem}
\begin{proof}
For $0 \le i \le k \le d$, and subspace $I,K \subseteq_s \F_q^d$ with $dim(I) 
= i, dim(K) = k$, define matrix $ \overline{W_{ik}}$ over $\R$ as	$ 
\overline{W_{ik}}(I,K)  = 1 $ if  $I \cap K = \{0\} $ and $0$ 
otherwise. This matrix has dimension $\gaussian{d}{i}{q} \times
\gaussian{d}{k}{q}$. 

Consider the submatrix
$M_i$ of $M$ with rows and columns indexed by subspaces of dimension exactly
$i$. Observe that $\overline{W_{ii}} = J - M_i$ where $J$ is 
an all ones matrix of appropriate order. These matrices are well-studied 
(see~\cite{FW86}). Closed form expressions for eigenvalues are  computed 
in~\cite{D76,LW12} and the eigenvalues are known to be non-zero. Hence for $0 
\le i \le d/2$ the matrix $\overline{W_{ii}}$ has rank $\gaussian{d}{i}{q}$.
Since $\overline{W_{ii}} = J - M_i$, $\rank{M_i} \ge \rank{\overline{W_{ii}}}- 1$.
This shows that $\rank{M} \ge \rank{M_i} = \gaussian{d}{i}{q}$ for all $i$ such
that $2i \le d$. Choosing $i=d/2$ gives $\rank{M} \ge \gaussian{d}{d/2}{q} -1 \ge
q^{\frac{d^2}{4}}-1$.
\end{proof}

We now present two approaches for showing lower bounds on $\upd(f)$ - one using intersection families of vector spaces and the other using rectangle arguments on $M_f$. 

\noindent
\textbf{Lower Bound for $\upd(\calP_d)$ using intersecting families of vector spaces :} To prove a lower bound on $\upd(\calP_d)$ we define a
matrix $N$ from a projective assignment with
intersection dimension $1$ for  $\calP_d$,
such that it is equal to $(q-1)M$. Let $D = \upd(\calP_d)$. We
first show that $\rank{N}$ is at most $1+\gaussian{D}{1}{q}$. 
Then by Lemma~\ref{th:subspace-graph-rank} we get that $\rank{N}$ is at least
$q^{\frac{d^2}{4}}$. 
Let $\calG=\set{G_1,\dots,G_m}$, $\calH=\set{H_1,\dots,H_m}$ be the
subspace assignment with intersection dimension $1$ realizing $\calP_d$ 
with dimension $D$.

\begin{lem}
\label{cl:matrix-poly-rank-ub}
	For any polynomial $p$ in $q^x$ of degree $s$, with
	matrix $N$ of order $|\calG| \times |\calH|$ defined as $N[G_r,H_t] = p(\dim(G_r \cap H_t))$  with $G_r \in \calG$, $H_t
	\in \calH$, then $\rank{N} \le \sum_{i=0}^s \gaussian{D}{i}{q}$
\end{lem}

\begin{proof}
This proof is inspired by the proof in \cite{FG85} of a similar claim where a non-bipartite version of this lemma is proved.
To begin with, note that $p$  is a degree $s$ polynomial in $q^x$, and hence can be written as a linear combination of
polynomials $p_i = \gaussian{x}{i}{q}, 0 \leq i \leq s$. Let the linear
combination be given by $p(x) = \sum_{i=0}^s \alpha_i p_i(x)$. For $0 \leq
i \leq s$ define a matrix $N_i$ with rows and columns indexed respectively by
$\calG$, $\calH$ defined as $N_i[G_r,H_s] =
p_i(\dim{G_r \cap H_s})$. By definition of $N_i$, $N = \sum_{i \in [s]}
\alpha_i N_i$.

To bound the rank of $N_i$'s we introduce the following families of
inclusion matrices. For any $j \in [D]$, consider two matrices $\Gamma_j$
corresponding to $\calG$ and $\Delta_j$ corresponding to $\calH$ defined as 
$\Gamma_j(G,I) = 1$ if $\dim(I) = j, G \in \calG, I \subseteq_s G$ and $0$ 
otherwise. $\Delta_j(H,I) = 1$ if $\dim(I) = j, H \in \calH, I \subseteq_s H$ and 
$0$ otherwise.
Note that rank of the these matrices are upper bounded by the number of
columns which is $\gaussian{D}{j}{q}$.
We claim that for any $i \in \{0,1,\ldots,s \}$, $\rank{N_i} \leq \gaussian{D}{i}
{q}$. This completes the proof since $N = \sum_{i \in [s]} \alpha_i N_i$.

To prove the claim, let $\calF_i$ denote the set of all $i$ dimensional subspace of 
$\F_q^D$.  We show that $N_i =
	\Gamma_i\Delta_i^T$. Hence $\rank{N_i} \le
	\min\set{\rank{\Gamma_i },\rank{\Delta_i}}
	\le \gaussian{D}{i}{q}$. For $(G_r, H_t) \in \calG \times \calH$,
\revsay{Changed eqn formatting to improve readability}
	\begin{eqnarray*} 
		\Gamma_i\Delta_i^T(G_r,H_t) 
	&=& \sum_{I  \in \calF_i}  \Gamma_i(G_r,I) \Delta_i^T(I,H_t) \\
	&=& \sum_{I  \in \calF_i}  \Gamma_i(G_r,I) \Delta_i(H_t,I)  \\
	&=& \sum_{I  \in \calF_i}   [I \subseteq_s G_r] \wedge [I \subseteq_s
	H_t] \\
	&=&  \sum_{I \in \calF_i} [I \subseteq_s G_r \cap H_t] \\
	&=& \gaussian{\dim(G_r \cap H_t)}{i}{q} = N_i(G_r,H_t)
\end{eqnarray*}
\end{proof}
We apply Lemma~\ref{cl:matrix-poly-rank-ub} on $N$ defined via $p(x) = q^x-1$
with $s=1$, to get $q^{d^2/4} \le \gaussian{d}{d/2}{q} \le 1+
\gaussian{D}{1}{q} = 1 + (q^D-1)/(q-1)$. 
By definition, $\rank{N} = \rank{M}$. 
This gives that $D = \Omega(d^2)$ and proves Theorem~\ref{thm:pd-lb-main}.

\noindent
\textbf{Lower Bound for $\upd(\calP_d)$ from Rectangle Arguments : }
We now give an alternate proof of for Theorem~\ref{thm:pd-lb-main} using combinatorial rectangle arguments. 
\begin{lem}
\label{cl:int-dim-1-lb}
	For $f:\set{0,1}^n\times \set{0,1}^n \to \set{0,1}$ with $M_f$ 
    denoting the bipartite adjacency matrix of $G_f$, $rank_{\R}(M_f) \le
	q^{O(\upd_\F(f))}$ where $\F$ is a finite field of size $q$.
\end{lem}
\begin{proof}
	Let $\phi$ be a subspace assignment realizing $f$ of dimension $d$ with intersection dimension 1. Let
	$S(v)$ for $v \in \F_q^d$ denote $\set{(a,b) \in \set{0,1}^n \times
	\set{0,1}^n ~|~ \phi(a) \cap \phi(b) =
	\vsp{v}}$. Also let $M_v$ denote the matrix representation of $S(v)$.
	That is, $M_v(a,b) = 1 \iff (a,b) \in S(v)$. Consider all $1$
	dimensional subspaces which appear as intersection space for
	some input $(x,y)$. Fix a basis vector for each space and let $T$
	denote the collection of basis vectors of all the intersection spaces. Note
	that for any $(x,y) \in f^{-1}(1)$, there is a unique $v \in \F_q^d$
	(up to scalar multiples) such that $(x,y) \in S(v)$ for otherwise
	intersection dimension exceeds $1$.
	Then $M_f = \sum_{v \in T} M_v$. Now, $rank(M_f) \le \sum_{v \in T} 
	rank(M_v)$. Since $rank(M_v) =	1$, $rank(M_f) \le |T|$.  The fact that 
	the number of $1$ dimensional spaces in $\F^d$ can be at most $\frac{q^d-1}{q-1}$ 
	completes the proof. Note that the rank of $M_f$ can be over any field (we choose $\R$).
\end{proof}
We get an immediate corollary. Any function $f$, such that the adjacency 
matrix of $M_f$ of the bipartite graph $G_f$ is of full rank $2^n$ over some field 
must have $\upd(f) = \Omega(n)$. There are several Boolean functions with this property, well-studied in the context of communication 
complexity (see textbook~\cite{NK97}). Hence, we have
for $f \in \set{\IP_n, \EQ_n, \INEQ_n, \DISJ_n, \PAL^q_n}$,  $\upd_\F(f)$ is  $\Omega(n)$ for any finite field $\F$.

For arguing about $\PAL^q_n$, it can be observed that the graph is strongly regular (as $q \equiv 1 \mod 4$)  and hence the adjacency
matrix has full rank over $\R$~\cite{Bol01}. Except for $\PAL^q_n$, all the above functions have $O(n)$ sized deterministic branching
programs computing them and hence the \pdrl theorem (Theorem~\ref{thm:pdrl}) gives that $\upd$ for these functions (except $\PAL^q_n$) are
$O(n)$ and hence the above lower bound is indeed tight.

From Lemma~\ref{th:subspace-graph-rank}, it follows that the function $\SI_d$ 
also has rank $2^{\Omega(d^2)}$.  To see this, it suffices to observe 
that $\calP_d$ appears as an induced subgraph 
in the bipartite realization of $\SI_d$. Thus, $\upd(\SI_d)$ is 
$\Omega(d^2)$. We proved in Theorem~\ref{thm:subspace-graph-pd} that 
$\pdim(\SI_d) = d$. This establishes a quadratic gap between the two 
parameters.  This completes the proof of Theorem~\ref{thm:pd-lb-main}.

Let $\D(f)$ denote the deterministic communication complexity of the Boolean function $f$.
We observe that the rectangle argument used in the proof of Lemma~\ref{cl:int-dim-1-lb} is similar to the matrix rank based lower bound arguments for communication complexity.
This yields the Proposition~\ref{prop:upd-cc}. If $\upd(f) \le D$, the assignment also gives a partitioning of the $1$s in $M_f$ into at most $\frac{q^D-1}{q-1}$ $1$-rectangles. However, it is unclear whether this immediately gives a similar partition of $0$s into $0$-rectangles as well. Notice that if $\D(f) \le d$, there are at most $2^d$ monochromatic rectangles (counting both $0$-rectangles and $1$-rectangles) that cover the entire matrix. However, our proof does not exploit this difference.

\begin{proposition}
\label{prop:upd-cc}
For a Boolean function $f:\{0,1\}^n \times \{0,1\}^n \to \{0,1\}$ and a finite field $\F$, $\upd_\F(f) \le 2^{\D(f)} \textrm{ and } \D(f)  \le (\pd_\F(f))^2 \log |\F|$
\end{proposition}
\begin{aproof}{Proof of \cref{prop:upd-cc}}{app:upd-cc}
We give a proof of the first inequality.
Any deterministic communication protocol computing $f$ of cost $\D(f)$,  
partitions $M_f$ into  $k$ rectangles where $k \le 2^{\D(f)}$ 
rectangles. Define $f_i:\{0,1\}^n \times \{0,1\}^n \to \{0,1\}$ for 
each rectangle $R_i$ $i \in [k]$, such that $f_i(x,y) = 1$ iff $(x,y) 
\in R_i$. Note that $\upd_\F(f_i) = 1$ and $f =  \lor_{i=1}^k f_i$. 
For any $(x,y) \in \set{0,1}^n \times \set{0,1}^n$ if $f(x,y)=1$, there 
is exactly one $i\in [k]$ where $f_i(x,y)=1$. Hence for each $j \in 
[k], j \neq i$, the intersection vector corresponding 
to the edge $(x,y)$ in the assignment of $f_j$ is trivial. 
Hence the assignment obtained by applying \cref{lem:OrLemma}, to $f_1,
\lor f_2 \lor \ldots f_k$ will have the property that for any $(x,y)$ 
with $f(x,y) = 1$, the intersection dimension is $1$. Hence $\upd_\F(f) 
\le k \le 2^{\D(f)}$. To prove the second inequality, consider the protocol where Alice sends the 
subspace associated with her input as a $\pd_\F(f) \times \pd_\F(f)$ matrix. Bob then checks if this subspace intersects with his own subspace and sends $1$ if it does so and sends $0$ otherwise.
\end{aproof}
An immediate consequence of \cref{prop:upd-cc} is that all symmetric functions $f$ on $2n$ bits have have projective dimension $O(n)$. 
Note that the first inequality is tight, up to constant factors in the exponent. To see this, consider the function $f:\zon \times \zon \to \zo$ whose
$\pd_\F(f) = \Omega(2^{n/2})$ \cite[Proposition 1]{PR92} and note that $\D(f)$ for any $f$ is at most $n$. Tightness of second inequality is witnessed by $\SI_d$ since by Lemma~\ref{th:subspace-graph-rank}, $\D(\SI_d) = \Omega(d^2)$ while $\pd(\SI_d) = d$.

\section{Bitwise Decomposable Projective Dimension}
The restriction of intersection dimension being 1, although potentially
useful for lower bounds for branching program size, does not 
capture the branching program size exactly. 
We start the section by demonstrating a function where the gap is exponential.
We show the existence of a Boolean function $f$ such that the size of the optimal branching program computing it is very high but has a very small projective assignment with intersection dimension 1 for any balanced partition of the input.

\begin{proposition}\label{thm:upd-failure}(Implicit in Remark 1.30 of \cite{juknatext})
There exist a function $f:\set{0,1}^n \times \set{0,1}^n$ that requires size $\Omega(\frac{2^n}{n})$ for any branching program computing $f$ but the $\upd(f) \le O(n)$ for any balanced partitioning of the input into two parts.
\end{proposition}
\begin{aproof}{Proof of \cref{thm:upd-failure}}{app:upd-failure}
Consider the function $\EQ_n$. The graph $G_{\EQ_n}(U,V,E)$ with $U=V=N$ is a perfect matching where $N = \zon$. Relabel the vertices in $U$ of this graph to produce a family of $\mathcal{G}$ of $N!$ different labeled graphs.
Let $\mathcal{F}$ be the set of Boolean functions whose corresponding 
graph is in $\mathcal{G}$ (or equivalently $\mathcal{F}$ of $N!$ different 
functions). Let $t$ be the smallest number such that any function in  
$\mathcal{F}$ can be computed by a branching program of size at most $t$. 
The number of branching programs of size $\le t$ (bounded by $O(t^t)$ \cite{juknatext}) forms 
an upper bound on $|\mathcal{F}|$. Thus, $2^{O(t\log t)} \ge N!$, and 
hence $t$ is  $\Omega \left(\frac{2^n}{n} \right )$. Hence there must 
exist a function $f \in \mathcal{F}$ such that $\upd(f) = \upd(\EQ_n) \le 
n$ but $\bpsize(f)$ is $\Omega \left(\frac{2^n}{n} \right )$ for this 
partition.

We now argue upper bound for $\upd(f)$ for any balanced partition. Consider the function $f_{\pi}$ obtained by a permutation $\pi \in S_N$ on the $U$ part of $\EQ_n$ graph. Consider a partition $\Pi$ of $[2 n]$. Let $G_{\EQ_n}^{\Pi}, G_{f_{\pi}}^{\Pi}$ be the corresponding bipartite graphs (and $\EQ_n^\Pi$ and $f_\pi^{\Pi}$ be the corresponding functions) with respect to the partition $\Pi$, of $\EQ_n$ and $f_{\pi}$ respectively.

We claim that  $\upd(G_{\EQ_n}^{\Pi}) = \upd(G^{\Pi}_{f_{\pi}})$. 
  By definition for any $(u,v) \in \set{0,1}^n \times \set{0,1}^n$, $f_{\pi} (u,v) = \EQ_n (\pi^{-1}(u),v)$. Also, let $(u', v')$ be the corresponding
  inputs according to the partition $\Pi$ of $[2 n]$. That is 
  $f^{\Pi}_{\pi}(u',v') = f_{\pi}(u,v) = \EQ_n (\pi^{-1}(u),v)$. Let $x = \pi^{-1}(u)$ and $y = v$. Observe that, for $(x, y)
  \in \{ 0, 1 \}^n \times \{ 0, 1 \}^n$ there is unique $(x', y')$
  corresponding to it. Hence $f^{\Pi}_{\pi} (u', v') = \EQ_n (\pi^{-1}(u),v) = \EQ_n^\Pi(x', y')$. Thus for any input $(u', v')$ of $f^{\Pi}_{\pi}$
  there is unique input $(x',y')$ of $\EQ_n^\Pi$ obtained via the above
  procedure. Thus, from the $\upd$ assignment for $\EQ^{\Pi}_n$ we can get a
  $\upd$ assignment for $f^{\Pi}_{\pi}$. Observing that Theorem~\ref{thm:pdrl-upd} holds for any partition $\Pi$ of the input, we get a $\upd$ assignment for $\EQ^{\Pi}_n$.
\end{aproof}

The above proposition can be shown by adapting the 
counting argument presented in 
Remark 1.30 of \cite{juknatext}. 

\subsection{A Characterization for Branching Program Size}
\label{sec:bpsize-char}
Motivated by strong properties observed in Proposition~\ref{prop:pdrl-props},  we make the following definition.

\begin{definition}[\textbf{Bitwise Decomposable Projective Dimension}]\label{def:bitpdim}
	Let $f$ be a Boolean function on $2n$ bits and $G_f$ be its bipartite
	realization. The bipartite graph $G_f(X,Y,E)$ is said to have \emph{bit
	projective dimension}, $\bpdim(G) \le d$, if there exists a 
	collection of subspaces of $\F_2^d$ denoted $\calC = \{U_i^a\}_{i \in
	[n], a \in \{0,1\}}$ and $\calD = \{V_j^b\}_{j \in [n], b \in
	\{0,1\}}$ where a projective assignment $\phi$ is obtained by
	associating  subspace $U_i^a$ with a bit assignment $x_i = a$ and
	$V_j^b$ with $y_j = b$ satisfying the conditions listed below.
	\begin{enumerate}
	\item \label{prop:bit-assgn}
		 for all $(x,y) \in \set{0,1}^n \times \set{0,1}^n$, $\phi(x) =
		\vspan{1\le i \le n}{U_i^{x_i}}$, $\phi(y)= \vspan{1\le j \le
		n}{V_j^{y_j}}$ and $f$ is realized by $\phi$.
	\item \label{prop:diff-std-basis}
	Let $S = \{e_i - e_j \mid i,j \in [d], i \ne j\}$. For any $W \in \calC \cup \calD$, $\exists S' \subseteq S$ such that $W = \vsp{S'}$. 
	 \item\label{def:prop-disj} 
	 		for any $S_1,S_2 \subseteq ([n] \times \set{0,1})$ such that $S_1 \cap S_2 = \phi$,
            $\vspan{(i,a) \in S_1}{U_i^a}\cap \vspan{(j,b) \in S_2}{U_j^b}=\set{0}$. Same property must hold for subspaces in $\calD$.
\end{enumerate}

\end{definition}

We show that the new parameter bitwise decomposable projective dimension ($\bpdim$) tightly characterizes the branching program size, up to constants in the exponent. 

\begin{lem}\label{bpdim-pdrl}
Suppose $f:\set{0,1}^n \times \set{0,1}^n \to \set{0,1}$ has deterministic
branching program of size $s$ then $\bpdim(f) \le 6s$
\end{lem}

\begin{aproof}{Proof of \cref{bpdim-pdrl}}{app:bpdim-pdrl}
The subspace assignment obtained by applying 
(\cref{thm:pdrl-appendix}) on an arbitrary
branching program need not satisfy Property~\ref{def:prop-disj} because there
can be a vertex $z$ that has two edges incident on it reading different
variables from the same partition. To avoid this, we subdivide every edge. We
show that this transformation is sufficient to get a $\bpdim$ assignment. We
now give a full proof.
\begin{figure}[htp!]
	\centering
	\includegraphics[scale=0.9]{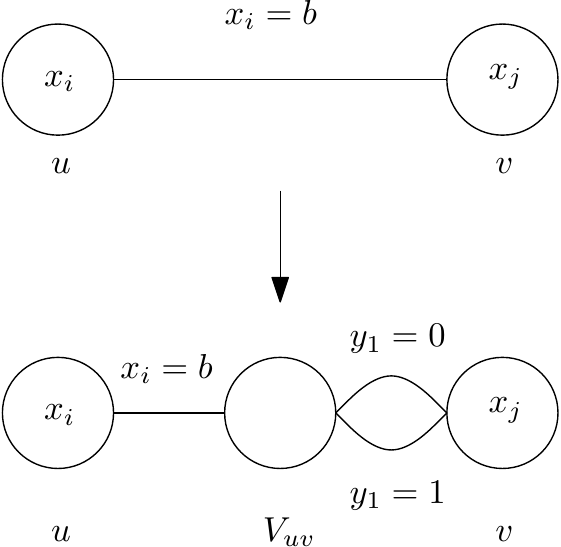}
	\caption{Edge modification}
	\label{fig:split}
\end{figure}
Let $B$ be a deterministic branching program computing $f$. Denote the first
$n$ variables of $f$ as $x$ and the rest as $y$. We first
apply \pdrl transformation on $B$ to obtain a branching program $B'$ computing
$f$. We note that $|V(B')|=|V(B)|$.  Obtain $B''$ from $B'$ by subdividing
every edge $(u,v)$ checking a variable $x_i=b$ from partition $x$ to get three
edges $(u,V_{uv})$ checking $x_i=b$ and add two edges between 
$(V_{uv},v)$ one which checks $y_1=0$  and another which checks $y_1=1$ 
(see Figure~\ref{fig:split}).

Clearly the transformation
does not change the function computed by the branching program. Since we are
taking every edge of the branching program $B'$ and introducing two more
edges, the total number of edges in $B'$ is $3|E(B')|$. Since $B'$ is a
deterministic
branching program, every vertex $v \in B'$ has out degree at most $2$
and at least $1$ for every node except sink node. Hence $|E(B')|\le 2(|V(B')|)$. 
Along with $|E(B'')|=3|E(B')|$, we get $|E(B'')|\le 6(|V(B')|)= 6(|V(B)|)$. Now label
every vertex of $B''$ with standard basis vectors as it is done in \pdrl Theorem
(\cref{thm:pdrl-appendix}). Let
$\phi$ be projective assignment obtained from $B''$ via
\pdrl theorem. We claim that $\phi$ satisfies all the requirements of
$\bpdim(f)$. 
\begin{enumerate}
  \item Since $\phi$ is obtained via \pdrl it captures adjacencies of $G_f$.
	  Hence property~\ref{prop:bit-assgn} holds.
	  Property~\ref{prop:diff-std-basis} is satisfied by \pdrl
	  assignment. (See appendix~\ref{app:pdrl-proof})
  \item The standard basis vector $e_u$ corresponding to vertex $u$ appears
	  only in edges incident on $u$ in \pdrl assignment. For any edge
	  $(u,v)$ querying a variable $x_i=b$ the other edges incident to $v$
	  must query variables from $y$. All the edges incident on
	  $u$, except $(u,v)$ must also query variables from $y$.
	  Otherwise, there is an edge $(w,u)$ which queries a variable $x_j$
	  and our transformation would have subdivided the edge. Hence
	  $e_u,e_v$ belongs only to $H_{x_i=b}$ amongst  $\set{H_{x_i=b}}_{i
		  \in [n], b \in \set{0,1}}$. This implies
		  Property~\ref{def:prop-disj}.
\end{enumerate}
\end{aproof}

We show that given a $\bpdim$ assignment for a function $f$, we can construct a branching program computing $f$.
\begin{theorem}[Theorem~\ref{thm:bitpdim-converse} restated]
For a Boolean function $f:\set{0,1}^n \times \set{0,1}^n \to \set{0,1}$ with $\bpdim(f) \le d$, there exists a deterministic branching program computing $f$ of size $d^c$ for some absolute constant $c$.
\end{theorem}
\begin{proof}
Consider the subspace associated with the variables $\calC, \calD$ 
of the $\bpdim$ assignment as the advice string. These can be specified by 
a list of $n$ basis matrices each of size $d^2$. Note that for any any $f$ which has a polynomial sized branching program, $d = \bpdim(f)$ is at most $poly(n)$, and hence the advice string is $poly(n)$ sized and depends only
on $n$.

We construct a deterministic branching program computing $f$ as follows. On 
input $x,y$, from the basis matrices in $\calC, \calD$, construct an undirected
graph\footnote{Note that this is not a deterministic branching program.} $G^*$ with all standard basis vectors in $\calC, \calD$ as vertices and add 
an edge between two vertices $u$, $v$ if $e_u-e_v \in U_i^{x_i}$ or $e_u-e_v \in V_j^{y_j}$ for $i,j \in [n]$. 
For input $x,y$, $f(x,y) = 1$ iff $G^*$ 
has a cycle. 
To see this, let $C = C_1 \cup C_2$ be a cycle in $G^*$ where $C_1$ consists of edges from 
basis matrices in $\calC$ and $C_2$ contain edges from basis matrices in 
$\calD$. Note that if one of $C_1$ or $C_2$ is empty then there is a cycle 
consisting only of vectors from $\calC$ which implies a linear dependence 
among vectors in $\calC$. But this contradicts Property~\ref{def:prop-disj} of 
$\bpdim$ assignment. Hence both $C_1$ and $C_2$ are non-empty.

Then, it must be that $\sum_{(u,v) \in C_1} e_u - e_v + \sum_{(w,z) \in C_2} 
e_w-e_z = 0$. Hence $\sum_{(u,v) \in C_1} e_u - e_v = -\sum_{(w,z) \in C_2} 
e_w-e_z $. Hence we get a vector in the intersection which gives $f(x,y) = 1$. Note that if $f(x,y) = 1$, then clearly there is a non-zero intersection 
vector. If we express this vector in terms of basis, we get a cycle in $G^*$. 

Hence, to check if $f$ evaluates to $1$, it suffices check if there is a cycle in $G^*$ which is solvable in $\L$ using Reingold's algorithm~\cite{Rei08}. The log-space algorithm can also be converted to an equivalent branching program of size $n^c$ for a constant
$c$.

 We can improve the constant $c$ to $3+\epsilon$. We achieve this using the well known random walk based $\RL$ algorithm for reachability~\cite{AKLLR79}, amplifying the error and suitably fixing the random bits to achieve a non-uniform branching program of size $d^{3+ \epsilon}$.

The $\RL$ algorithm requires to store $\log d$ bits to remember the current vertex while doing the random walk and 
another $\log d$ bits to store the next vertex in the walk. It performs a walk of length $4d^3$ and answers correctly with probability of $1/2$ \cite{MU05}. Amplifying the error does not incur any extra space as the algorithm has a one-sided error and it never errs when it accepts. This gives a probabilistic Turing machine using $2\log d + 1$ work space. By amplifying the success probability, we can obtain a choice of a random bits which works for all inputs of a fixed length. %This de-randomization is similar to the de-randomization showing $\BPP \subseteq \Sigma_2$ (see \cite[Theorem 7.14]{text:AB}).\revsay{Corrected some typos and "$BPP \in L$"}
The conversion of this machine to a branching program will incur storing of the head index position of the work tape and input tape position which incur an additional $\log\log d + \log d$ space.  Hence overall space is $3\log d + \log \log d = (3+\epsilon) \log d$ for small fixed $\epsilon > 0$, thus proving that $c \leq 3 + \epsilon$.

\end{proof}

Assuming $\CeL \not\subseteq 
\L/\poly$, the function $\SI_d$ (a language which is hard for $\CeL$ under Turing reductions) cannot be 
computed by deterministic branching programs of polynomial size.
\begin{atheoremproof}{Proof of $\CeL$ Hardness of $\calP_d$}{app:pd-hardness}
\begin{proposition}\label{th:pd-hardness}
	The function family $\{\SI_d\}_{d \ge 0}$ is hard for $\CeL$ via
	logspace Turing reductions. Moreover, the negation of $\{\SI_d\}_{d \ge 0}$ is in $\L^{\CeL}$ (and hence in $\NC^2$).
\end{proposition}
\begin{proof}
We start with the following claim.
\begin{claim}[Corollary 2.3 of \cite{ABO99}]
\label{prop:det-manip}
	Fix an $n \in \N$. 
	There exists a logspace computable function $g:\F^{n \times n} \to
	\F^{n \times n}$ such that for any
	matrix $M$ over $\F^{n \times n}$,
		$det(M) = 0 \implies rank(g(M)) = n$ and 
		$det(M) \ne 0 \implies rank(g(M)) = n - 1$
\end{claim}

Consider the language $L = \set{(M_1, M_2) ~|~  \rowspan(M_1) \cap \rowspan(M_2)
\ne \{0\}, M_1, M_2 \in \F^{d \times d} }$. The reduction is as follows. 
Given an $M \in \F^{d \times d}$, apply $g$ (defined in Claim~\ref{prop:det-manip}) on $M$ to get $N$, and define for
$1 \le i \le d$, $H^i = (M_1^i, M_2^i)$ where $M_1^i$ is the matrix consisting
of $i^{th}$ row of $N$ repeated $n$ times and $M_2^i$ as same as 
$N$ with $i^{th}$ row replaced by all $0$ vectors. For each $1 \le i \le
d$, we make oracle query to $L$ checking if $H^i \in L$ and if all answers are
no, reject otherwise accept.

We now argue the correctness of the reduction.  Suppose $det(M)$ is $0$, then
$N = g(M)$ (by~\cref{prop:det-manip}) must have full rank. Hence for all $1
\le i \le d$, $\rowspan(M_1^i)$ and $\rowspan(M_2^i)$ does not intersect. If
$det(M) \ne 0$, then $N = g(M)$ (by~\cref{prop:det-manip}) must have a
linearly dependent column and hence there is some $i$ for which
$\rowspan(M_1^i)$ and $\rowspan(M_2^i)$ is non-zero. Also the overall
reduction runs in logspace as $g$ is logspace computable.

The upper bound follows by observing that given two $d \times d$ matrices $M_1$ and $M_2$, their individual ranks $r_1$ and $r_2$ can be computed in $\L^{\CeL}$~\cite{ABO99}. Consider the matrix $M$ of size $d \times 2d$ by adjoining $M_1$ and $M_2$. It follows that the $rowspace(M_1) \cap rowspace(M_2) \ne \phi$ if and only if $\rank{M} < r_1+r_2$. The latter condition can also be tested using a query to $\CeL$ oracle.
\end{proof}
\end{atheoremproof}

\subsection{Lower Bounds for Bitwise Decomposable Projective dimension}

From the results of the previous section, it follows that size lower bounds 
for branching programs do imply lower bounds for bitwise decomposable 
projective dimension as well. 
As mentioned in the introduction, the lower bounds 
that Theorem~\ref{thm:bitpdim-converse} can give for bitwise decomposable 
projective dimension are only known to be sub-linear.

To prove super-linear lower bounds for bitwise decomposable projective dimension, we show that Nechiporuk's method~\cite{Nec66} can be adapted to our linear algebraic framework (thus proving Theorem~\ref{thm:bitpdim-lb} from the introduction). The overall idea is the 
following: given a function $f$ and a $\bpdim$ assignment $\phi$,
consider the restriction of $f$ denoted $f_\rho$ where $\rho$ fixes
all variables except the ones in $T_i$ to $0$ or $1$ where $T_i$ is some subset of
variables in the left partition. For different restrictions $\rho$,
we are guaranteed to get at least $c_i(f)$ different functions. We show that 
for each restriction $\rho$, we can obtain an assignment from $\phi$ realizing $f_\rho$.
Hence the number of different $\bpdim$ assignments for $\rho$ 
restricted to $T_i$ is at least the number of sub functions of $f$ which is at 
least $c_i(f)$. Let $d_i$ be the ambient dimension of the assignment when 
restricted to $T_i$. By using the structure of $\bpdim$ assignment, we count 
the number of assignments possible and use this relation to get a lower 
bound on $d_i$. Now repeating the argument with disjoint $T_i$, and by 
observing that the subspaces associated with $T_i$s are disjoint, we 
get a lower bound on $d$ as $d = \sum_i d_i$.

\begin{theorem}
\label{th:pdrl-count}
	For a Boolean function $f:\set{0,1}^n \times \set{0,1}^n \to \set{0,1}$
	on $2n$ variables, let $T_1, \ldots, T_{m}$ are
	partition of variables to $m$ blocks of size
	$r_i$ on the first $n$ variables. Let $c_i(f)$ be the number of
	distinct sub functions of $f$ when restricted to $T_i$, then
	$\bpdim(f) \ge \sum_{i=1}^m \frac{\log c_i(f)}{\log (\log c_i(f))}$
\end{theorem}
\begin{proof}
	Let $(x,y)$ denote the $2n$ input variables of $f$ and
	$\rho:\set{x_1,\dots,x_n,y_1,\dots,y_n} \to \set{0,1,*}$ be a map that
	leaves only variables in $T_i$ unfixed. Let $\phi$ be a
	$\bpdim$ assignment realizing $f$ and let $G_f(X,Y,Z)$ denote the
	bipartite realization of $f$. Let $\calC=\set{U_i^a}_{i\in [n], a
	\in \set{0,1}}, \calD = \{V_j^b\}_{j \in [n], b \in \{0,1\}}$ be the
	associated collection of subspaces.  Let $\rho$ be a restriction that
	does not make $f_\rho$ a constant and $(x,y) \in \set{0,1}^n \times
	\set{0,1}^n$ which agrees with $\rho$. We use $x,y$ to denote both
	variables as well as assignment.
	 From now on, we fix an $i$ and a partition $T_i$.

	Define 
		$L = \vspan{i\in [n], \rho(i)\neq *}{U_i^{\rho(i)}}$ and
		$R = \vspan{j\in [n]}{V_j^{\rho(n+j)}}$.
	For any $x \in \set{0,1}^n$ that agrees with $\rho$ on the first $n$
	bits, define 
	$Z^x = \vspan{j \in T_i}{U_j^{x_j}}$
	Note that for any $(x,y)$, which agrees with $\rho$, has $\phi(x) =
	L+Z^x$ and $\phi(y)=R$.
 For any $f_{\rho_1} \not\equiv f_{\rho_2}$, $G_{f_{\rho_1}}
	\neq G_{f_{\rho_2}}$. Hence the number of $\bpdim$ assignments is at
	least the number of different sub functions. 
	We need to give a $\bpdim$ assignment for 
	$G_{f_\rho}(V_1,V_2,E)$ where $V_1=\set{x \mid x \text{ agrees with }
	\rho}$, $V_2=\set{y}$ where $y = \rho_{[n+1,\dots,2n]}$ and $E = \set{
	(x,y) | x \in V_1, y \in V_2, f(x,y)=1}$. We use the following property to come up with such an assignment.
\begin{property}
         \label{property:splittingOfIntVector}
	 Let $\rho$ be a restriction which does not make the function
         $f$ constant and which fixes all the variables $y_1,\ldots,y_n$.
         For all such $\rho$ and  $\forall x,y \in \set{0,1}^n$ which
	 agrees with $\rho$, any non-zero $w \in \phi(x) \cap \phi(y)$, 
	 where $w=u+v$ with $u \in L$ and $v \in Z^x$ must satisfy $v \neq
	 \vec{0}$.
       \end{property}
\begin{proof}
       Let there exists an intersection vector $w \in (L+Z^x) \cap R$ with
       $w=u+v$, $u\in L$ and $v \in Z^x$ and $v=\vec{0}$. Since $\vec{0} \in
       Z^{\hat{x}}$ for any $\hat{x}$, $w=u+\vec{0}$ is in $L + Z^{\hat{x}}$
       and $R$. Thus the function after restriction $\rho$ is a constant. This
       contradicts the choice of $\rho$.
       \end{proof}
The assignment $\psi_{\rho}$	for $G_{f_\rho}$ defined as :
  $\psi_{\rho}(x) = Z^x$ and $\psi_{\rho}(y) = \vspan{x \in V_1}{\Pi_{Z^x}\left( R \cap \left( L + Z^x \right)\right)}$
Note that for $(x,y) \in V_1 \times V_2$, $f_\rho(x) = f(x,y)$. Following
claim shows that $\psi_{\rho}$ realize $f_\rho$.  
\begin{claim}
For any $(x,y) \in V_1 \times V_2$,
$f(x,y)=1$ if and only if $\psi_{\rho}(x) \cap \psi_{\rho}(y) \neq \set{0}$.
\end{claim}

\begin{proof}
For any $(x,y) \in X \times Y$, $\phi(x) \cap \phi(y) \neq \set{0}$ if
and only if $f(x,y)=1$. Since $V_1 \subseteq X$ and $V_2 \subseteq Y$, it suffices to prove : $\forall (x,y) \in V_1 \times V_2$, $\psi_{\rho}(x) \cap \psi_{\rho}(y)
\neq \set{0} \iff \phi(x) \cap \phi(y) \ne \set{0}$.

We first prove that $\psi_{\rho}(x) \cap \psi_{\rho}(y) \neq \set{0}$ implies $\phi(x) \cap
\phi(y) \neq \set{0}$. Let $v$ be a non-zero vector in $\psi_{\rho}(x) \cap \psi_{\rho}(y)$.
By definition of $\psi_{\rho}(x)$, $v \in Z^x$.  By definition of $\psi_{\rho}(y)$, there
exists a non-empty $J \subseteq V_1$ such that $v =\sum_{\hat{x} \in
J}v_{\hat{x}}$ where $v_{\hat{x}} \in Z^{\hat{x}}$. Also for every $\hat{x}
\in J$, there exists a $u_{\hat{x}} \in L$ such that $w_{\hat{x}} =
u_{\hat{x}}+ v_{\hat{x}}$ and $w_{\hat{x}} \in R$. Define $u$ to be
$\sum_{\hat{x} \in J} u_{\hat{x}}$. Since each $u_{\hat{x}}$ is in $L$, $u$ is
also in $L$. Hence $w = u+v$ is in $L+Z^x$.
Substituting $u$ with $\sum_{\hat{x} \in J} u_{\hat{x}}$ and $v$ with 
$\sum_{\hat{x} \in J} v_{\hat{x}}$ we get that $w = \sum_{\hat{x} \in J}
u_{\hat{x}}+v_{\hat{x}}=\sum_{\hat{x} \in J} w_{\hat{x}}$. Since each
$w_{\hat{x}} \in R$, $w  \in R$. Hence $w \in R \cap (L + Z^x)$ and $w$ is
non-zero as $J$ is non-empty.

Now we prove that $\phi(x) \cap \phi(y) \neq \set{0}$ implies $\psi_{\rho}(x)
\cap \psi_{\rho}(y) \neq \set{0}$. Let $w$ be non zero vector in $\phi(x)
\cap \phi(y)$ with $w=u+v$ where $u \in L$ and $v \in Z^x$. By
Property~\ref{property:splittingOfIntVector} we have $v \ne
\vec{0}$. By definition $v \in \psi_{\rho}(y)$. Along with $v \in Z^x$,
we get $\psi_{\rho}(x) \cap \psi_{\rho}(y) \ne \set{0}$.
\end{proof}
Let $Z=\vspan{j \in T_i}{U_j^0 + U_j^1}$. We now prove that subspace 
assignment on the only vertex in the right partition of $G_{\rho}$ which 
is $\vspan{x \in V_1}{\Pi_{Z^x}(R)}$ is indeed $\Pi_Z(R)$. 
\begin{claim}
	$ \Pi_Z(R) = \vspan{x \in V_1}{\Pi_{Z^x}(R)}$ \label{cl:proj-z}
\end{claim}
\begin{proof}
We show $ \vspan{x \in V_1}{\Pi_{Z^x}(R)} \subseteq \Pi_Z(R)$.
Note that $\vspan{x \in V_1}{\Pi_{Z^x}(R)} = \vspan{x \in V_1,w
\in R}{\Pi_{Z^x}(w)}$. For an arbitrary $x \in V_1$ and $w \in R$, let
$v = \Pi_{Z^x}(w)$. By definition of $Z^x$ and the fact that
$\set{U_i^b}_{i \in [n],b \in \set{0,1}}$ are disjoint, $\Pi_{Z^x}(w) =
+_{i \in [n], \rho(i)=*} \Pi_{U_i^{x_i}}(w)$. As $Z=\vspan{j \in T_i}
{U_j^0 + U_j^1}$, every $\Pi_{U_i^{x_i}}(w) \in \Pi_Z(R)$. Hence the
span is also in $\Pi_Z(R)$.

Now we show that $\Pi_Z(R) \subseteq \vspan{x \in V_1} {\Pi_{Z^x}(R)}$. Let
$T_i=\set{i_1,\dots,i_k}$. For $1\leq j \leq k$ define $x^{j}$ to be $x+e^{j}$
where $x \in \set{0,1}^n$ agrees with $\rho$ and for any index $i \in [n]$
with $\rho(i)=*$, $x_i=0$ and $e^j \in \set{0,1}^n$ is $0$ at every index
other than $i_j$. Note that for any $j_1 \neq j_2, j_1,j_2 \in T_i$,
$Z^{x^{j_1}} \cap Z^{x^{j_2}} = \set{0}$ by Property~\ref{def:prop-disj} of Definition~ \ref{def:bitpdim}) Also
note that $\vspan{j \in T_i}{Z^{x^j}}= \vspan{j \in T_i}{U_j^{x_j}} =Z$. Hence, $\Pi_Z(R)  = \vspan{j \in
T_i}{\Pi_{Z^{x^j}}(R)}$. But $\vspan{j \in T_i}{\Pi_{Z^{x^j}}(R)} \subseteq
\vspan{x \in V_1}{\Pi_{Z^x}(R)}$. %Hence the proof.
\end{proof}

For any $\rho$, which fixes all variables outside $T_i$, $Z$ is the
same. And since there is only one vertex on the right partition, for 
different $\rho, \rho'$, $\Pi_Z(R_\rho) = \Pi_Z(R_{\rho'})$ implies 
$\psi_{\rho} = \psi_{\rho'}$. Hence to count the number of different 
$\psi_{\rho}$'s  for different $f_\rho$'s it is enough to count the number of 
different $\Pi_Z(R)$. To do so, we claim the following property on $\Pi_Z(R)$. 
\begin{property} \label{prop:BasisSetOfStandardAndDiff}
Let $S=\set{e_u-e_v| e_u-e_v \in Z}$. Then there exists a subset $S'$ of $S$ 
such that all the vectors in $S'$ are linearly independent and
$\Pi_Z(R) = \vsp{S'}$.
\end{property}
\begin{proof}
By the property of the $\bpdim$ assignment, $\forall i\in [n]$ and $\forall b \in
\set{0,1}$, $V_i^b=\vsp{F_i^b}$ where $F_i^b$ is a collection of difference of
standard basis vectors. Recall that $R = \vspan{j\in
  [n]}{V_j^{\rho(n+j)}}$.
Let $F =\set{ (e_u-e_v) \mid e_u-e_v \in F_j^{\rho(n+j)}, j\in [n]}$.
 Since projections are linear maps and the fact that
$F_j^{\rho(n+j)}$ spans $V_j^{\rho(n+j)}$ we get that,
  $\Pi_Z(R) = \vspan{w \in F} {\Pi_Z(w)}$.
Since $Z$ is also a span of difference of standard basis vectors,
$\Pi_Z(e_u-e_v)$ is one of $\vec{0}$, $e_u-e_w$ or $e_w-e_v$ where $e_w$
is some standard basis vector in $Z$. Let $S'' = \cup_{e_u-e_v \in F}
\Pi_Z(e_u-e_v)$. Hence $S'' \subseteq S$. Clearly, $\vspan{e_u-e_v \in
  S''}{e_u-e_v} = \Pi_Z(R)$.  Choose $S'$ as a linear independent subset of
$S''$.
\end{proof}

Property~\ref{prop:BasisSetOfStandardAndDiff} along with the fact that
$\Pi_Z(R)$ is a subspace of $Z$, gives us that the number of different
$\Pi_Z(R)$ is upper bounded by number of different subsets $S'$ of $S$
such that $|S'|=d_i$ where $d_i=\dim(Z)$. As $S'$ is a set of
difference of standard basis vectors from $Z$,
$|S'| \leq \binom{d_i}{2}$. Thus the number of different such $S'$ are
at most $\sum_{k=0}^{d_i}\binom{d_i^2}{k}=2^{O(d_i \log d_i)}$.

Hence the number of restrictions $\rho$ (that leaves $T_i$ unfixed)
and leading to different $f_\rho$ is at most $2^{O(d_i \log
d_i)}$. But the number of such restrictions $\rho$ is at least
$c_i(f)$. Hence $2^{O(d_i \log d_i)} \ge c_i(f)$ giving $d_i =
\Omega \left (\frac{\log c_i(f)}{\log (\log c_i(f))}\right)$. Using
$d = \sum_i d_i$ completes the proof.
\end{proof}
Theorem~\ref{th:pdrl-count} gives a super linear lower bound for Element Distinctness function.
From a manuscript by Beame et.al, (\cite{BGMS16}, see also \cite{juknatext}, Chapter 1), we have $c_i(ED_n) \ge 2^{n/2}/n$. Hence applying this count to Theorem~\ref{th:pdrl-count}, we get that $d \ge \Omega\left(\frac{n}{\log n} \cdot \frac{n}{\log n}\right) = \Omega\left (\frac{n^2}{(\log n)^2} \right )$. 

Now we apply this to our context. 
To get a lower bound using framework described above it is enough to
count the number of sub-functions of $\SI_d$.

\begin{lem}\label{Pdcount}
For any $i \in [d]$, there are $2^{\Omega(d^2)}$ different restrictions
$\rho$ of $\SI_d$ which fixes all entries other than $i$th row of the
$d \times d$ matrix in the left partition.
\end{lem}
\begin{proof}
Fix any $i \in [d]$. 
Let $S$ be a subspace of $\field_2^d$. Define $\rho_S$ to be
$\SI_d(\mathbf{A},B)$ where $B$ is a matrix whose rowspace is
$S$. And $\mathbf{A}$ is the matrix whose all but $i$th row is $0$'s
and $i$th row consists of variables $(x_{i_1},\dots,x_{i_n})$. Thus
for any $v \in \set{0,1}^d$, rowspace of $\mathbf{A}(x)$ is 
$\vsp{v}$.

We claim that for any $S,S' \subseteq_S \field_2^d$ where $S \neq S'$,
$\left( \SI_d  \right)_{\rho_S} \not\equiv \left( \SI_d
\right)_{\rho_S'}$. By definition $\left( \SI_d  \right)_{\rho_S}
\equiv \SI_d(\mathbf{A},B)$ and $\left( \SI_d  \right)_{\rho_S'}
\equiv \SI_d(\mathbf{A},B')$ where $B$ and $B'$ are matrices whose
rowspaces are $S$ and $S'$ respectively. Since $S \neq S'$ there is at
least one vector $v \in \field_2^d$ such that it belongs to only one
of $S,S'$. Without loss of generality let that subspace be $S$. Then
$\SI_d(\mathbf{A}(v),B)=1$ as $v \in S$ where as
$\SI_d(\mathbf{A}(v),B')=0$ as $v \not\in S'$. Hence the number of
different restrictions is at least number of different subspaces of
$\field_2^d$ which is $2^{\Omega(d^2)}$. Hence the proof.
\end{proof}
This completes the proof of Theorem~\ref{thm:bitpdim-lb} from the introduction. 
	This implies that for $\SI_d$, the branching program size
	lower bound is $\Omega \left( \frac{d^2}{\log d} \times d \right) =
	\Omega \left( \frac{d^3}{\log d}\right) = \Omega \left
	(\frac{n^{1.5}}{\log n }\right)$ where $n=2d^2$ is the number of input bits of $\SI_d$.

\section{Standard Variants of Projective Dimension}
In this section, we study two stringent variants of projective dimension for which exponential lower bounds and exact characterizations can be derived. Although these measure do not correspond to restrictions on branching programs, they illuminate essential nature of the general measure. We define the measures and show their characterizations in terms of well-studied graph theoretic parameters. 

\begin{definition}[\textbf{Standard Projective Dimension}]
  \label{defn:spd}
  A Boolean function $f : \set{0,1}^n \times \set{0,1}^n \to
  \set{0,1}$ with the corresponding bipartite graph $G(U,V,E)$ is said
  to have standard projective dimension (denoted by $\spd(f)$) $d$ over field $\field$, if $d$ is the smallest possible dimension for which there exists a vector space $K$ of dimension $d$ over $\field$ with a map $\phi$ assigning subspaces of $K$ to $U \cup V$  such that
  \begin{itemize}
  \item for all $(u,v) \in U \times V$, $\phi(u) \cap \phi(v) \neq \set{0}$ if and only if $(u,v) \in E$.
  \item $u \in U \cup V$, $\phi(u)$ is spanned by a subset of standard basis vectors in $K$.
\end{itemize}     

\end{definition}

In addition to the above constraints, if the assignment satisfies the property that for all $(u,v) \in U \times V$, $\dim\left(\phi(u) \cap \phi(v)\right) \le 1$, we say that the \textit{standard projective dimension is with intersection dimension $1$}, denoted by $\uspd(f)$.
We make some easy observations about the definition itself.

For $N \times N$ bipartite graph $G$ with $m$ edges, consider the 
assignment of standard basis vectors to each of the edges and for any $u \in 
U \cup V$, $\phi(u)$ is the span of the basis vectors assigned to the edges 
incident on $u$. Moreover, the intersection dimension in this case is $1$. 
Hence for any $G$, $\spd(G) \le \uspd(G) \le m$.

Even though 
$\pd(G) \le \spd(G)$, there are graphs for which the gap is exponential. For example, consider the bipartite realization $G$ of $\EQ_n$ with $N = 2^n$. We know $\pd(G) = \Theta(\log N)$ but $\spd(G) 
\ge N$ since each of the vertices associated with the matched edges cannot share any basis vector with vertices in other matched edges.  Hence dimension must be at least $N$.
We show that standard projective dimension of bipartite $G$ is same as that of biclique cover number.
\begin{definition}[Biclique cover number]
For a graph $G$, a collection of complete bipartite graphs defined on $V(G)$
is said to cover $G$ if every edge in $G$ is present in some complete
bipartite graph of the collection.  The size of the smallest collection of
bipartite graph which covers $G$ is its biclique cover number (denoted by
$\bc(G)$). If in addition, we insist that bicliques must be {edge-disjoint}, the parameter is known as \textit{biclique partition number} denoted by $\bp(G)$.
\end{definition}

\begin{theorem}[Restatement of Theorem~\ref{th:spdEqbc}]
For any Boolean function $f$,  $\bc(G_f) = \spd(G_f)$ and $\uspd(G_f) = \bp(G_f)$.
\end{theorem}
\begin{aproof}{Proof of \cref{th:spdEqbc}}{app:spd-uspd}
($\spd(f) \le \bc(G_f)$) Let $G = G_f$, $t = \bc(G)$ and $A_1.\ldots,
A_t$ be a bipartite cover for $G$. For a vertex $v \in V(G)$, let
$I_v = \left\{ e_i ~|~ v \in A_i \right\}$. We claim that $\{I_v\}_{v \in
V(G)}$ is a valid standard projective assignment. Suppose $I_u \cap I_v \ne
\emptyset$, then there exists an $i$ such that $u,v \in A_i$ and $(u,v) \in
E(A_i)$. Hence $(u,v) \in E(G)$. Also if $(u,v) \in E(G)$, then $\exists~i$
s.t. $(u,v) \in E(A_i)$. By definition of $I_u,I_v$,  $e_i \in I_u \cap I_v$
giving $I_u \cap I_v \ne \emptyset$.

($\bc(G_f) \le \spd(G_f)$) Let $G = G_f, t = \spd(G)$ and $\{I_u\}_{u \in V(G)}$
be the subsets assigned. Consider $G_i = \left\{(u,v) ~|~ i \in I_u
\text{ and } i \in I_v  \right\}$ for $i \in \{1,\ldots,t\}$. 
We claim that the collection of $G_i$ forms a valid bipartite cover of $G$. If
$(u,v) \in E(G)$, we have $I_u \cap I_v \ne \emptyset$. Hence there exists an
$i \in I_u \cap I_v$ and $(u,v) \in E(G_i)$. If $(u,v) \in E(G_i)$ for some
$i$, then $i \in I_u$ and $i \in I_v$ implying $I_u \cap I_v \ne \emptyset$.
This gives that $(u,v) \in E(G)$ from the definition of standard assignment.

($\bp(G_f) \leq \uspd(G_f)$) Let $\phi$ be the
intersection dimension one standard assignment of ambient dimension
$d$ of $f$. For every $e_i \in \F^d$, define the set $C_i = \set{(x,y)
  \mid \phi(x,y)=e_i}$. We claim that $\family{C} = \set{C_i}_{i \in
  [d]}$ is a bipartite partition of $G_f$. Every $C_i$ thus defined is
a biclique, because if $\phi(x,y) = e_i$ then that implies $e_i \in
\phi(x)$ and $e_i \in \phi(y)$.
 Note that for every $(x,y)
\in G_f$, there exists a unique $i \in [d]$ such that
$\phi(x,y)=e_i$. Hence any $(x,y) \in G_f$ belongs to exactly one of
the sets $C_i$ thus implying that $C_i$'s are edge disjoint biclique
covers. Note that any $(x,y) \not\in G_f$ do not belong to any of
$C_i$'s as $\phi(x,y) = \set{0}$.

($\uspd(G_f) \leq \bp(G_f)$) Let
$\family{C}=\set{C_i}_{i \in [d]}$ where $d=\bp(G_f)$ be a biclique
partition cover. We give a standard assignment $\phi$ for $G_f$
defined as follows. For any $x$, $\phi(x) = \vsp{e_i \mid \exists y,
  (x,y) \in C_i}$. By definition $\phi$ is a standard assignment. We
just need to prove that $(x,y) \in G_f$ if and only $\phi(x,y) \neq
\set{0}$ and $\dim{\phi(x,y)}=1$. To prove this we would once again
employ the rectangle property of bicliques, that is if $(x,y') \in
C_i$ and $(x',y) \in C_i$ then so is $(x,y)$.
First we will argue that if there an intersection then it is dimension
$1$. Recall that intersection of two standard subspaces is a standard
subspace. Suppose there is exists $(x,y)$ with $\dim{\phi(x,y)}
>1$. Let $e_j,e_k$ be any two standard intersection vectors in
$\phi(x,y)$. By construction and rectangle property of bicliques, we
get that $(x,y) \in C_j$ and $(x,y) \in C_k$ contradicting the
disjoint cover property. Hence for any $(x,y)$, $\dim{\phi(x,y)} \leq 1$. If $(x,y) \not\in G_f$, then there does not exist an $i$, $(x,y) \in C_i$. But if $\phi(x,y) = e_i$ for some $i \in [d]$, then that implies by rectangle property of bicliques that $(x,y) \in C_i$, a contradiction.
\end{aproof}

\section{Discussion \& Conclusion}

In this paper we studied variants of projective dimension of graphs with improved connection to branching programs. We showed lower bounds for these measures indicating the weakness and of each of the variants. 
A pictorial representation of all parameters is shown in \cref{fig:params}.
\begin{figure} 
\centering
\includegraphics[scale=0.9]{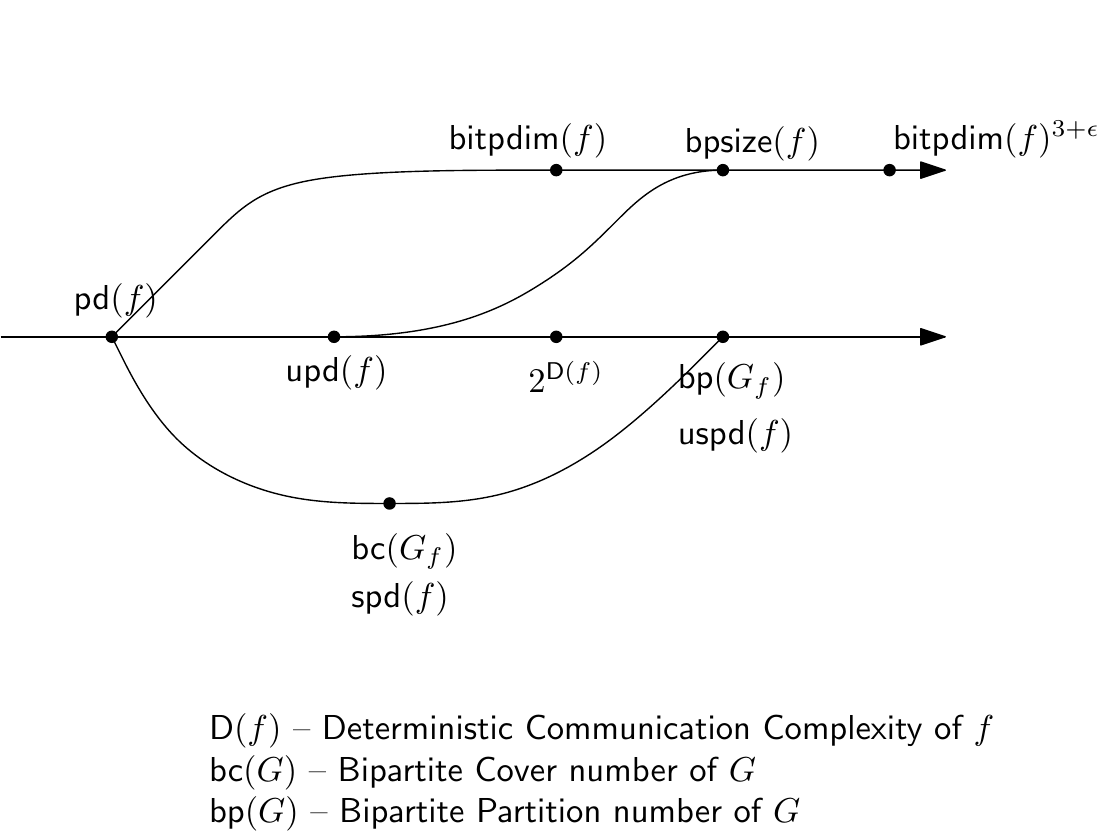}
\caption{Parameters considered in this work and their relations}
\label{fig:params}
\end{figure}

An immediate question that arises from our work is whether $\Omega(d^2)$ lower bound on $\upd(\calP_d)$ is tight. In this direction, since we have established a gap between $\upd(\calP_d)$ and $\pd(\calP_d)$, it is natural to study how $\pd$ and $\upd$ behave under composition of functions, in order to amplify this gap. 

In another direction, we believe that the $\Omega(d^2)$ lower bound on $\upd(\calP_d)$ is not tight. It is natural to study composition of functions to improve this gap. 

The subspace counting based lower bounds for $\bpdim$ that we proved are tight for functions like $\ED_n$. However, observe that under standard complexity theoretic assumptions the $\bpdim$ assignment for $\calP_d$ is not tight. Hence it might be possible to use the specific linear algebraic properties of $\calP_d$ to improve the $\bpdim$ lower bound we obtained for $\calP_d$.
\revsay{We have commented some of the future directions here. Should we uncomment any of them ?}
\vspace{2mm}

\noindent{\bf Acknowledgements:} The authors would like to thank the anonymous reviewers for several suggestions which improved the readability of the paper and specifically for pointing out that the proof of Proposition~\ref{thm:upd-failure} follows from Remark~1.3 in \cite{juknatext}. The authors would also like to thank Noam Nisan for pointing out that the a random walk based algorithm for detecting cycles can improve the the constant in Theorem~\ref{thm:bitpdim-converse} to $3+\epsilon$.\revsay{just "to $3+\epsilon$". 7 is not needed I think}

\bibliographystyle{plain}
\bibliography{references}
\appendix
\section{Proof of \pdrl theorem}
\label{app:pdrl-proof}
In this section, we reproduce the proof of the projective dimension upper bound in terms of branching program size. The proof is originally due to \cite{PR92}, but we supply the details which are essential for the observations that we make.

A deterministic branching program is a directed acyclic graph $G$ with distinct start ($V_0$), accept ($V_+$) and reject ($V_-$) nodes. Accept and reject nodes
 have fan-out zero and are called \emph{sink} nodes. Vertices of the DAG,
 except sink nodes are labeled by variables and have two
outgoing edges, one labeled $0$ and the other labeled $1$. For a vertex
labeled $x_i$, if input gives it a value $b\in \{0,1\}$, then the edge 
labeled $b$ incident to $x_i$ is said to be \emph{closed} and the other edge 
is \emph{open}.
A branching program is said to accept an input $x$ if and only if there is a
path from $V_0$ to $V_+$ along the closed edges in the DAG. A branching
program is said to compute an $f:\zon \to \zo$, if for all $x \in \zon$, $f(x) =
1$ iff branching program accepts $x$.

\begin{theorem}
\label{thm:pdrl-appendix}
Let $f:\{0,1\}^{2n} \to \{0,1\}$ be computed by a branching program $\calB$ 
of size $s$. Let $G_f$ be the bipartite realization of $f$, with respect to any partition of $[2n]$ into two parts and $\F$ be any 
arbitrary field. Then, $\pd_\F(G_F) \le s$
\end{theorem}
\begin{proof}
It suffices to come up with a subspace assignment $\phi$ such that $G_f(P,Q,E)$ has a
projective representation in $\F$. Associate $u, v$ to be vertices in $P, Q$ respectively. In other words, $u$
corresponds to input variables $\{x_1,x_2,\ldots,x_n\}$ and $v$ corresponds to
$\{x_{n+1}, \ldots, x_{2n}\}$ (corresponding to the given partition). By the acceptance property of branching program
$\calB$,
$f(u\circ v) = 1 \iff \exists \text{ a path from $V_0$ to accept in $\calB$}$.
Since vertices in $G_f$ corresponds to strings in $\{0,1\}^n$, it suffices to
give an assignment $\phi$ such that
\begin{equation}
\text{$\exists$ a path from start to accept in $\calB$} \iff \text{ Basis of }
\phi(u), \phi(v) \text{ are linearly dependent}
\label{eq:pdrl-proof}
\end{equation}
We first assign vectors to vertices of the branching program and 
then use it to come up with a subspace assignment.

Suppose there is a path from $v_0$ to accept in $\calB$. A simple possible 
way to have dependence is to have sum of the vectors assigned to the edges 
of the path telescoping to zero. This can be achieved in the following way.
\begin{enumerate}
\item Modify $\calB$ by adding a new start vertex labeled with a variable
from the other partition from which $v_0$ got its label. For example, if 
$V_0$ is labeled with any of $x_1,x_2,\ldots,x_n$, the new vertex gets its 
label from $\{x_{n+1}, \ldots,x_{2n}\}$ and vice-versa.  
\label{mod:new-vert} Connect both outgoing edges labeled $0,1$ to $V_0$. 
\item Merge the accept node with the new start node. Let $\calC$ be the
resultant graph which is no longer acyclic. Assign standard basis vectors to
each vertex in $\calC$.
\item Assign to each edge $(u,v)$ the vector $e_u-e_v$.
\end{enumerate}
Now, the subspace assignment to a vertex $v \in V(G_f)$ is to take span of all
vectors assigned to closed edges on the input $v$. If there are no closed edges, we
assign the zero subspace. With the above modification, cycles in the graph 
would lead to telescoping of difference vectors (along the cycle edges) to sum 
to zero.

\begin{figure}[htp!]
\includegraphics[scale=0.8]{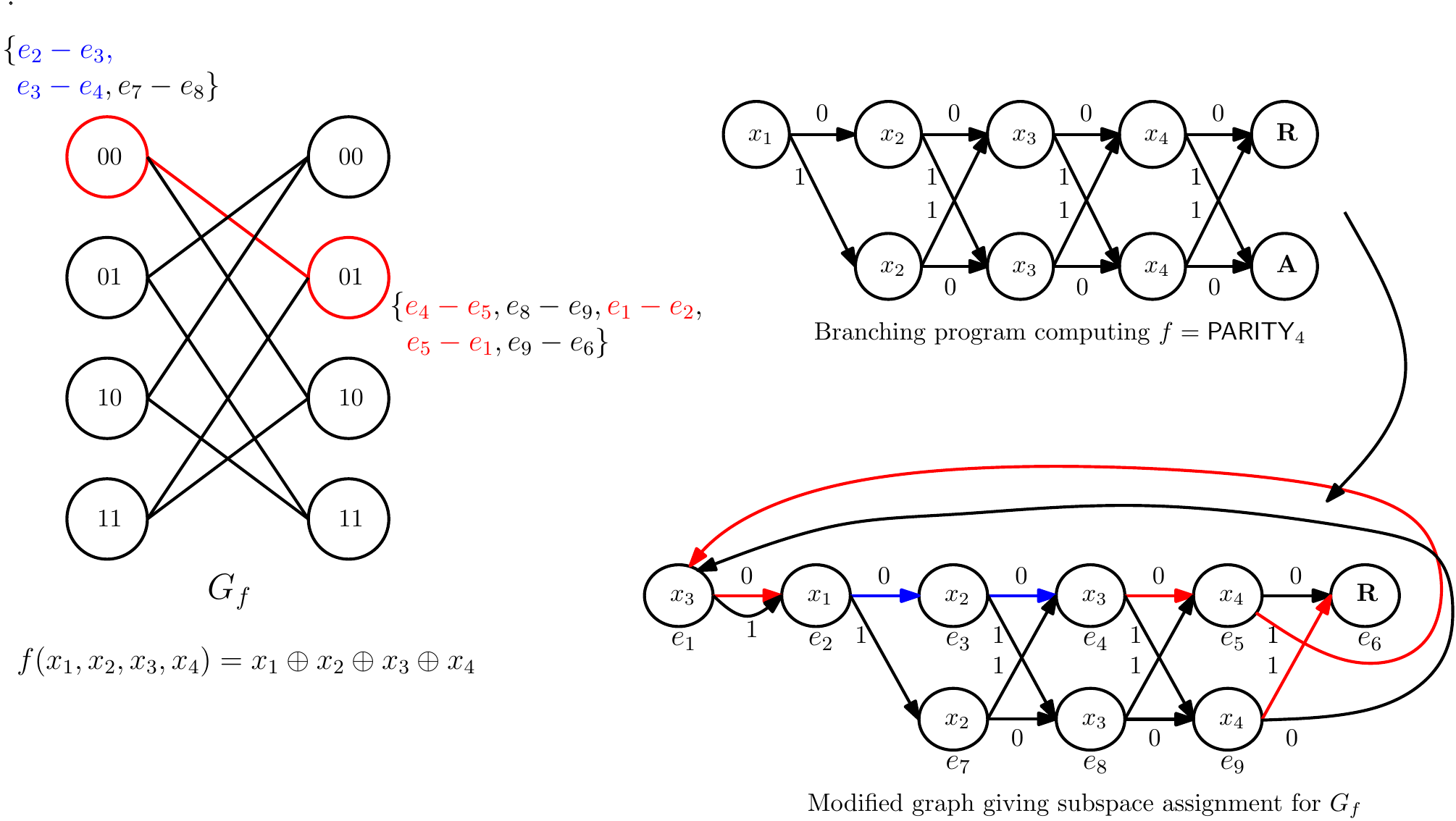}
\caption{\pdrl Theorem applied to a branching program computing $\mathsf{PARITY}_4$}
\label{fig:pdrl-proof}
\end{figure}

Modification~(\ref{mod:new-vert}) is necessary as it is possible to have a 
cycle that does not contain any vertex labeled with $\{x_{n+1}, \ldots,
x_{2n}\}$. Then $\phi(v)$ will be just zero subspace and $\phi(u) \cap \phi(v)$
will be trivial when there is a cycle. It is to avoid this that we add a
vertex labeled with variable from the other partition.

To show that $\phi$ is a valid subspace assignment, it remains to show that
reverse implication of statement~\ref{eq:pdrl-proof} holds. Suppose for $(u,v)
\in E(G_f)$, $\phi(u), \phi(v)$ are linearly dependent. Hence there exists a non
trivial combination giving a zero sum. \[ \sum_{\substack{e \in E(\calC) \\ e =
(u,v)}} \lambda_e (e_u-e_v) = 0, ~~~\lambda_e \in \F ~\forall e \in
E(C)\]
Let $S$ be the non-empty set of edges such that $\lambda_e \ne 0$ and 
$V(S)$ be its set of vertices. Now for any vertex $u \in V(S)$ there must be 
at least two edges containing $u$ because with just a single edge
$\epsilon_u$, which being a basis vector and summing up to zero, must have a zero coefficient which
contradicts that fact that $ e\in S$. This shows that every vertex in $S$ has
a degree $\ge 2$ (in the undirected sense). Hence it must have an undirected
cycle.
\end{proof}
\cref{fig:pdrl-proof} shows the transformations done to the branching program as per the proof of \pdrl Theorem and subspace assignment obtained for $00$ and $01$. The intersection vector for $00$ and $01$ is highlighted in blue on the left partition and in red on the right partition. Notice that this intersection vector corresponds to two halves of a cycle starting from start vertex of the modified BP.
The subspace assignment for each of the vertices is listed in the
table below.
\begin{table}[htp!]
\centering
\begin{tabular}{|c|c||c|c|} \hline
$x_1x_2$ & Assignment & $x_3x_4$ & Assignment \\
\hline \hline
$00$ & $e_2-e_3$, $e_3-e_4$, $e_7-e_8$ &
$00$ & $e_4-e_5$, $e_8-e_9$, $e_1-e_2$, $e_5-e_6$, $e_9-e_1$ \\ \hline
$01$ & $e_2-e_3$, $e_3-e_8$, $e_7-e_4$ &
$01$ & $e_4-e_5$, $e_8-e_9$, $e_1-e_2$, $e_5-e_1$, $e_9-e_6$ \\\hline
$10$ & $e_2-e_7$, $e_3-e_4$, $e_7-e_8$ &
$10$ & $e_4-e_9$, $e_8-e_5$, $e_1-e_2$, $e_5-e_6$, $e_9-e_1$ \\\hline
$11$ & $e_2-e_7$, $e_3-e_8$, $e_7-e_4$ &
$11$ & $e_4-e_9$, $e_8-e_5$, $e_1-e_2$, $e_5-e_1$, $e_9-e_6$ \\ \hline
\end{tabular}
\caption{Subspace assignment for $\mathsf{PARITY_4}$ given by proof of \pdrl theorem}
\end{table}

\section{Bounds on the Gaussian Coefficients}
\label{app:subspace-count}
\begin{proposition}[Lemma 1 of \cite{KK08}]
	For integers, $k \ge 0, n \ge k$.
	\begin{align}
	q^{k(n-k)} \le \gaussian{n}{k}{q} < c_q q^{k(n-k)}
	\end{align}
	where $c_q = \prod_{j=1}^\infty \frac{1}{1-q^{-j}}$. Note that for all
	$q \ge 2$, $c_q \le c_2 = 3.462\ldots$
\end{proposition}
\begin{proof} 
	Note that since $n \ge k$, $q^n \ge q^k$, we have $\frac{q^n-t}{q^k-t}
	\ge \frac{q^n}{q^k}$ for any $0 \le t < q^k$. Hence the lower bound
	follows.

	For the upper bound, 
	\begin{align*}
	 \gaussian{n}{k}{q} & = \frac{(q^n-1)(q^n-q) \ldots
		 (q^n-q^{k-1})}{(q^k-1) (q^k-q) \ldots (q^k-q^{k-1})} \\
		 & = \frac{q^{nk}}{q^{k^2}} \left [
		 \frac{(1-q^{-n})(1-q^{-(n-1)}) \ldots (1-q^{-(n-k+1)})}
	 	{(1-q^{-1})(1-q^{-2})\ldots (1-q^{-(k-1)})}\right ]
	\end{align*}
	Numerator of the previous expression can be upper bounded by $q^{nk}$
	while denominator can be lower bounded by $q^{k^2}(c_q)^{-1}$. This completes the
	proof.
\end{proof}
\begin{remark}
	This shows that the total number of subspaces of an $n$ dimensional
	space is upper bounded by $2 c_q \sum_{i=0}^{n/2} q^{in} \le 2c_2
	q^{(n^2+n)/2}$.
\end{remark}

\section{Proof of Proposition~\ref{prop:dir-prod-prop}}
\label{app:andlemma}
\begin{proof}
For the reverse direction, suppose there is a non zero vector $w_1$ in $U_1
\cap V_1$ and a non zero vector $w_2$ in $U_2 \cap
V_2$, then $w_1^Tw_2 \in U_1 \otimes U_2$ and $w_1^Tw_2 \in V_1
\otimes V_2$. Hence $w=w_1^Tw_2 \in (U_1 \otimes U_2) \cap (V_1 \otimes
V_2)$.

For the forward direction, let $w$ be a non zero vector in $(U_1
\otimes U_2) \cap (V_1 \otimes V_2)$. Let $\{e_i\}_{i\in [k_1]}$ be the
set of basis vectors for $F^{k_1}$ and $\{\tilde e_j\}_{j\in [k_2]}$
be the set of basis vectors for $F^{k_2}$. Hence for some
$\lambda_{ij}, \mu_{ij} \in \F$, $w$ can be written as, $w = \sum_{i,j}
\lambda_{ij}e_i^T \tilde e_j = \sum_{i,j} \mu_{ij} e_i^T \tilde e_j $. Hence, 
$\sum_{i,j}(\lambda_{ij} - \mu_{ij}) e_i^T \tilde e_j = 0$.
By linear independence of tensor basis,  
\begin{equation}
\lambda_{ij} = \mu_{ij} ~\forall~(i,j) \in [k_1] \times [k_2] 
\label{eq:proof}
\end{equation}

Since $w$ is non-zero, there exists $i_1,j_1$ with $(i_1,j_1) \in [k_1] \times
[k_2]$ such that $\lambda_{i_1j_1} \ne 0$.
Applying equation~\ref{eq:proof}, we get $\mu_{i_1j_1} \ne 0$. Hence for
$(i_1,j_1)$, $\lambda_{i_1j_1}, \mu_{i_1j_1}$ are both non-zero. 

Hence it must be that $(U_1 \otimes U_2)$ and $(V_1 \otimes V_2)$ has the 
vector $e_{i_1}^T\tilde e_{j_1}$. So $e_{i_1}$ must be present in $U_1$ and
$V_1$ and $e_{j_1}$ must be present in $U_2$ and $V_2$ (if not,
$e_{i_1}^T\tilde e_{j_1}$ would not have appeared in the intersection). Hence
$U_1 \cap V_1 \ne \{0\}$ and $U_2 \cap V_2 \ne \{0\}$.
\end{proof}

\end{document}